\documentclass[conference]{IEEEtran}

\newcommand{\descr}[1]{\vspace{0.2cm} \noindent \textbf{#1}}

\usepackage{epsfig}
\usepackage{xcolor, colortbl}
\definecolor{citegreen}{HTML}{458B00}
\usepackage[hyphens]{url}
\usepackage{hyperref}
\hypersetup{
   colorlinks=true,
   citecolor=citegreen
}
\usepackage{graphicx}
\usepackage[caption=false,font=footnotesize]{subfig}
\usepackage{float}

\newcommand{\sample}[1]{\mathrel{{\leftarrow}\vcenter{\hbox{\scriptsize\rmfamily\upshape\ensuremath{{#1}}}}}}

\usepackage{amsmath}
\usepackage{amssymb}
\usepackage{xfrac}
\usepackage{mathtools}

\usepackage{multirow}
\usepackage{fancyhdr}
\usepackage[ruled, linesnumbered, noend]{algorithm2e}

\newtheorem{proposition}{Proposition}[section]
\newtheorem{remark}{Remark}[section]

\usepackage{adjustbox}
\usepackage{array}
\usepackage{lipsum}
\let\labelindent\relax
\usepackage{enumitem}
\newcolumntype{x}[1]{>{\centering\arraybackslash\hspace{0pt}}p{#1}}
\newcolumntype{R}[2]{%
    >{\adjustbox{angle=#1,lap=\width-(#2)}\bgroup}%
    l%
    <{\egroup}%
}

\graphicspath{ {images/} }

\newcommand{\blue}[1]{{#1}}

\begin{document}

\title{On the Resilience of Biometric Authentication Systems against Random Inputs}



\author{\IEEEauthorblockN{Benjamin Zi Hao Zhao}
\IEEEauthorblockA{University of New South Wales\\and Data61 CSIRO\\
benjamin.zhao@unsw.edu.au}
\and
\IEEEauthorblockN{Hassan Jameel Asghar}
\IEEEauthorblockA{Macquarie University\\ and Data61 CSIRO\\
hassan.asghar@mq.edu.au}
\and
\IEEEauthorblockN{Mohamed Ali Kaafar}
\IEEEauthorblockA{Macquarie University\\ and Data61 CSIRO\\
dali.kaafar@mq.edu.au}}





\IEEEoverridecommandlockouts
\makeatletter\def\@IEEEpubidpullup{6.5\baselineskip}\makeatother
\IEEEpubid{\parbox{\columnwidth}{
    Network and Distributed Systems Security (NDSS) Symposium 2020\\
    23-26 February 2020, San Diego, CA, USA\\
    ISBN 1-891562-61-4\\
    https://dx.doi.org/10.14722/ndss.2020.24210\\
    www.ndss-symposium.org
}
\hspace{\columnsep}\makebox[\columnwidth]{}}

\maketitle

\begin{abstract}
We assess the security of machine learning based biometric authentication systems against an attacker who submits uniform random inputs, either as feature vectors or raw inputs, in order to find an \emph{accepting sample} of a target user. The average false positive rate (FPR) of the system, i.e., the rate at which an impostor is incorrectly accepted as the legitimate user, may be interpreted as a measure of the success probability of such an attack. However, we show that the success rate is often higher than the FPR. In particular, for one reconstructed biometric system with an average FPR of 0.03, the success rate was as high as 0.78. This has implications for the security of the system, as an attacker with only the knowledge of the length of the feature space can impersonate the user with less than 2 attempts on average. We provide detailed analysis of why the attack is successful, and validate our results using four different biometric modalities and four different machine learning classifiers. Finally, we propose mitigation techniques that render such attacks ineffective, with little to no effect on the accuracy of the system.  
\end{abstract}




\section{Introduction}
Consider a machine learning model trained on some user's data accessible as a black-box API for biometric authentication. Given an input (a biometric sample), the model outputs a binary decision, i.e., accept or reject, as its prediction for whether the input belongs to the target user or not. Now imagine an attacker with access to the same API who has never observed the target user's inputs. The goal of the attacker is to impersonate the user by finding an \emph{accepting sample} (input). What is the success probability of such an attacker? 

Biometric authentication systems are generally based on either physiological biometrics such as fingerprints~\cite{maltoni2009handbook}, face~\cite{yi2014learning, schroff2015facenet}, and voice~\cite{nagrani2017voxceleb, chung2018voxceleb2}), or behavioral biometrics such as touch~\cite{mahbub2016active} and gait~\cite{xu2017keh}, the latter category generally used for continuous and implicit authentication of users. These systems are mostly based on machine learning: a binary classifier is trained on the target user's data (positive class) and a subset of data from other users (negative class). This process is used to validate the performance of the machine learning classifier and hence the biometric system~\cite{chauhan2017behaviocog, xu2017keh, curran2017one, huang2018breathlive, liu2018vocal, chen2017your, song2016eyeveri, chauhan2017breathprint, ho2017mini, crawford2017authentication}. The resulting proportion of negative samples (other users' data) successfully gaining access (when they should have been rejected) produces the false positive rate (FPR, also referred as False Acceptance Rate). The target user's model is also verified for their own samples, establishing the false reject rate (FRR). The parameters of the model can be adjusted to obtain the equal error rate (EER) at which point the FPR equals FRR.

Returning to our question, the FPR seems to be a good indicator of the success probability of finding an accepting sample. However, this implicitly assumes that the adversary is a human who submits samples using the same human computer interface as other users, e.g., a smartphone camera in case of face recognition.
When the model is accessible via an API the adversary has more freedom in choosing its probing samples. This may happen when the biometric service is hosted on the cloud (online setting) or within a secure enclave on the user's device (local setting). In particular, the attacker is free to sample uniform random inputs. It has previously been stated that the success probability of such an attack is exponentially small~\cite{pagnin2014leakage} or it can be derived from the FPR of the system~\cite{martinez-hill, vulnerable-face}.\footnote{We note that these observations are made for \emph{distance-based} authentication algorithms and not machine-learning model based algorithms. See Sections \ref{sec:distance_classifier} and \ref{sec:related} for more details.}

In this paper, we show that uniform random inputs are accepted by biometric systems with a probability that is often higher and independent of the FPR. Moreover, this applies to the setting where the API to the biometric system can be queried using feature vectors after processing raw input as well as at the raw input level. A simple toy example with a single feature can illustrate the reason for the efficacy of the attack. Suppose the feature is normalized within the interval $[0, 1]$. All of target user's samples (the positive class) lie in the interval $[0, 0.5)$ and the other users' samples (the negative class) lie in the interval $(0.5, 1]$. A ``classifier'' decides the decision boundary of $0.5$, resulting in identically zero FRR and FPR. However, a random sample has a 50\% chance of being accepted by the biometric system.\footnote{This example is an oversimplification. In practice the training data is almost never nicely separated between the two classes. Also, in higher dimensions one expects exponentially small volume covered by samples from the positive and negative classes as is explained in Section~\ref{sec:attack_def}.} The success of the attack shows that the FPR and FRR, metrics used for reporting the accuracy of the classifier, cannot alone be used as proxies for assessing the security of the biometric authentication system. 

\noindent Our main contributions are as follows:

\begin{itemize}[leftmargin=5mm]
    \item We theoretically and experimentally show that in machine learning-based biometric authentication systems, the \emph{acceptance region}, defined as the region in feature space where the feature vectors are accepted by the classifier, is significantly larger than the true positive region, i.e., the region where the target users samples lie. Moreover, this is true even in higher dimensions, where the true positive region tends to be exponentially small~\cite{data-science}.   
    \item As a consequence of the above, we show that an attacker who has access to a biometric system via a black-box feature vector API, can find an accepting sample by simply generating random inputs, at a rate which in many cases is higher than implicated by the FPR. For instance, the success probability of the attack is as high as 0.78 for one of the systems whose EER is only 0.03. The attack requires minimum knowledge of the system: the attacker only needs to know the length of the input feature vector, and permissible range of each feature value (if not normalized). 
    \item We show that the success rate of a random input attack can also be higher than FPR if the attacker can only access the API at the raw input level (before feature extraction). For instance, on one system with an EER of 0.05, the success rate was 0.12. We show that the exponentially small region spanned by these raw random inputs rarely overlaps with the true positive region of any user in the system, owing to the success probability of the attack. Once again the attack only requires minimum knowledge of the system, i.e., the range of values taken by each raw input.
  \item To analyze real-world applicability of the attack, we reconstruct four biometric authentication schemes. Two of them are physiological, i.e., face recognition~\cite{schroff2015facenet} and voice authentication~\cite{nagrani2017voxceleb}. The other two use behavioral traits, i.e., gait authentication~\cite{gafurov2006biometric}, and touch (swipes) authentication~\cite{mahbub2016active, touchalytics}. For each of these modalities, we use four different classifiers to construct the corresponding biometric system. The classifiers are linear support vector machines (SVM), radial SVM, random forests, and deep neural networks. For each of these systems, we ensure that our implementation has comparable performance to the reference. 
    \item Our experimental evaluations show that the average acceptance region is higher than the EER in 9 out of 16 authentication configurations (classifier-modality pairs), and only one in the remaining 7 has the (measured) average acceptance region of zero. Moreover, for some users this discrepancy is even higher.
    For example, in one user model (voice authentication using random forests) the success rate of the random (feature) input is 0.55, when the model's EER is only 0.03, consistent with the system average EER of 0.03.
\item We propose mitigation techniques for both the random feature vector and raw input attacks. For the former, we propose the inclusion of beta-distributed noise in the training data, which ``tightens'' the acceptance region around the true positive region. For the latter, we add feature vectors extracted from a sample of raw inputs in the training data. Both strategies have minimal impact on the FPR and TPR of the system. The mitigation strategy renders the acceptance region to virtually 0 for 6 of the 16 authentication configurations, and for 15 out of 16, makes it lower than FPR. \blue{For reproducibility, we have made our codebase public.\footnote{\blue{Our code is available at: \url{https://imathatguy.github.io/Acceptance-Region}}}} 
    
\end{itemize}

We note that a key difference in the use of machine learning in biometric authentication as compared to its use in other areas (e.g., predicting likelihood of diseases through a healthcare dataset) is that the system should only output its decision: accept or reject~\cite{garcia2018explainable}, and not the detailed confidence values, i.e., confidence of the accept or reject decision. This makes our setting different from \emph{membership inference} attacks where it is assumed that the model returns a prediction vector, where each element is the confidence (probability) that the associated class is the likely label of the input sample~\cite{shokri2017membership, salem2018ml}. In other words, less information is leaked in biometric authentication. Confidence vectors can potentially allow an adversary to find an \emph{accepting sample} by using a hill climbing approach~\cite{martinez-hill}, for instance. 

\section{Background and Threat Model}
\label{sec:background}

\subsection{Biometric Authentication Systems}
The use of machine learning for authentication is a binary classification problem.\footnote{We note that sometimes a \emph{discrimination model}~\cite{xu-soups} may also be considered where the goal is to identify the test sample as belonging to one of $n$ users registered in the system. Our focus is on the authentication model. \blue{Also, see Section~\ref{sec:discussion}.}} The positive class is the target user class, and the negative class is the class of one or more other users. The target user's data for training is obtained during the registration or enrollment phase. For the negative class, the usual process is to use the data of a subset of other users enrolled in the system \cite{han2003personal, chauhan2017behaviocog, xu2017keh, curran2017one, huang2018breathlive, liu2018vocal, chen2017your, song2016eyeveri, chauhan2017breathprint, ho2017mini, crawford2017authentication}. Following best machine learning practice, the data (from both classes) is split into a training and test set. The model is learned over the training set, and the performance of the classifier, its misclassification rate, is evaluated on the test set. 

A raw biometric sample is usually processed to extract relevant features such as fingerprint minutiae or frequency energy components of speech. This defines the feature space for classification. As noted earlier, the security of the biometric system is evaluated via the misclassification rates of the underlying classifier. Two types of error can arise. A type 1 error is when a positive sample (target user sample) has been erroneously classified as negative, which forms the false reject rate (FRR). Type 2 error occurs when a negative sample (from other users) has been misclassified as a positive sample, resulting in the false positive rate (FPR). By tuning the parameters of the classifier, an equal error rate (EER) can be determined which is the rate at which FRR equals FPR. One can also evaluate the performance of the classifier through the receiver operator characteristic (ROC) curve, which shows the full relationship between FRR and FPR as the classifier parameters are varied. 

Once a biometric system is set up, i.e., classifier trained, the system takes as input a biometric sample and outputs accept or reject. In a continuous authentication setting, where the user is continually being authenticated in the background, the biometric system requires a continuous stream of user raw inputs. It has been shown that in continuous authentication systems the performance improves if the decisions is made on the \emph{average} feature vector from a set of feature vectors~\cite{unobservable, touchalytics, chauhan2016gesture}.

\subsection{Biometric API: The Setting}
We consider the setting where the biometric system can be accessed via an API. More specifically, the API can be queried by submitting a biometric sample. The response is a binary accept/reject decision.\footnote{For continuous authentication systems, we assume that the decision is returned after a fixed number of one or more biometric samples.} The biometric system could be \emph{local}, in which case the system is implemented in a secure enclave (a trusted computing module), or \emph{cloud-based} (online), in which the decision part of the system resides in a remote server. We consider two types of APIs. The first type requires raw biometric samples, e.g., the input image in case of face recognition. The second type accepts a feature vector, implying that the feature extraction phase is carried out before the API query. This might be desirable for the following reasons. 

\begin{itemize}[leftmargin=5mm]
    \item Often the raw input is rather large. For instance, in case of face recognition, without compression, an image will need every pixel's RGB information to be sent to the server for feature extraction and authentication. In the case of an image of pixel size $60 \times 60$, this would require approximately 10.8 KB of data. If the feature extraction was offloaded to the user device, it would produce a 512 length feature embedding, which can take as little as 512 bytes. 
    This also applies to continuous authentication which inherently requires a continual stream of user raw inputs. But often decisions are only made on an \emph{average} of a set of feature vectors~\cite{unobservable, touchalytics, chauhan2016gesture}. In such systems, only sending the resultant extracted average feature vector to the cloud also reduces communication cost.
    \item Recent studies have shown that raw sensory inputs can often be used to track users~\cite{das2016tracking}. Thus, they convey more information than what is simply required for authentication. In this sense, extracting features at the client side serves as an \emph{information minimization} mechanism, only sending the relevant information (extracted feature vectors) to the server to minimize privacy loss.
    \item Since the machine learning algorithm only compares samples in the feature space, only the feature representation of the template is stored in the system. In this case, it makes sense to do feature extraction prior to querying the system. 
    
\end{itemize}
From now onwards, when referring to a biometric API we shall assume the \emph{feature vector} based API as the default. We shall explicitly state when the discourse changes to raw inputs. Figure~\ref{fig:ar-diagram} illustrates the two APIs.

\subsection{Threat Model and Assumptions}
We consider an adversary who has access to the API to a biometric system trained with the data of a target user whom the adversary wishes to impersonate. More specifically, the adversary wishes to find an \emph{accepting sample}, i.e., a feature vector for which the system returns ``accept.'' In the case of the raw input API, the adversary is assumed to be in search for a raw input that results in an accepting sample after feature extraction. We assume that the adversary has the means to bypass the end user interface, e.g., touchscreen or phone camera, and can thus \emph{programmatically} submit input samples. There are a number of ways in which this is possible. 

In the online setting, a mis-configured API may provide the adversary access to the authentication pipeline.
In the local setting, if the feature extractor is contained within a secure environment, raw sensory information must be passed to this protected feature extraction process. To achieve this an attacker may construct their own samples through OS utilities. An example is the Monkey Runner~\cite{monkeyrunner} 
on Android, a tool allowing developers to run a sequence of predefined inputs for product development and testing. Additionally, prior work~\cite{shwartz2017shattered} has shown the possibility of compromising the hardware contained within a mobile device, e.g., a compromised digitizer can inject additional touch events. 

Outside of literature, it is difficult to know the exact implementation of real-world systems. However, taking face recognition as an example, we believe our system architecture is similar to real world facial authentication schemes, drawing parallels to pending patent US15/864,232 \cite{son2018face}. 
Additionally there are API services dedicated to hosting different components of the facial recognition pipeline. Clarifai, for example, hosts machine learning models dedicated to the extraction of {face embeddings} within an uploaded image~\cite{clarifai-face}.
A developer could then use any number of Machine Learning as a Service (MLaaS) providers to perform the final authentication step, without needing to pay premiums associated with an end-to-end facial recognition product.

\noindent We make the following assumptions about the biometric API.

\begin{itemize}[leftmargin=5mm]
    \item The input feature length, i.e., the number of features used by the model, is public knowledge.
    \item Each feature in the feature space is min-max normalized. Thus, each feature takes value in the real interval $[0, 1]$. This is merely for convenience of analysis. Absent this, the attacker can still assume plausible universal bounds for all features in the feature space.
    \item The attacker knows the identifier related to the user, e.g., the username, he/she wishes to impersonate.
\end{itemize}

Beyond this we do not assume the attacker to have any knowledge of the underlying biometric system including the biometric modality, the classifier being used, the target user's past samples, or any other dataset which would allow the attacker to infer population distribution of the feature space of the given modality.

\begin{figure}[th!]
	\centering
    \includegraphics[width=0.80\columnwidth]{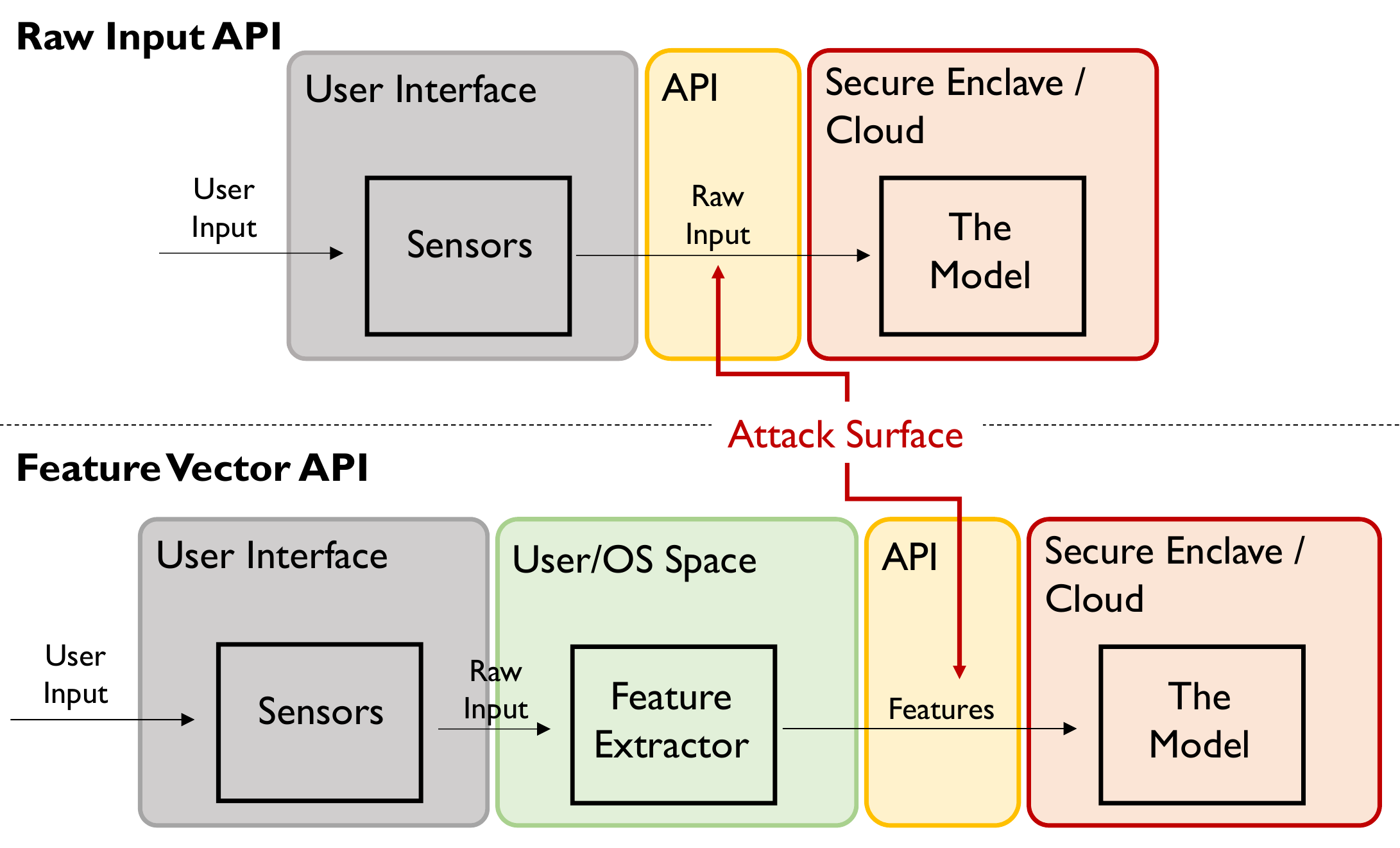}
    \vspace{-4mm}
  	\caption{The threat model and the two types of biometric API.}
 	\label{fig:ar-diagram}
\end{figure}

\section{Acceptance Region and Proposed Attack}
\label{sec:attack_def}

\subsection{Motivation and Attack Overview}
Given a feature space, machine learning classifiers learn the region where feature vectors are classified as positive features and the region where vectors are classified as negative features. We call the former, the acceptance region and the latter the rejection region. Even though the acceptance region is learnt through the data from the target user, it does not necessarily tightly surround the region covered by the target user's  samples. Leaving aside the fact that this is desirable so as to not make the model ``overfitted'' to the training data, this implies that even vectors that do not follow the distribution of the target user's samples, may be accepted. In fact these vectors may bare no resemblance to any positive or negative samples from the dataset. Consider a toy example, where the feature space consists of only two vectors. The two-dimensional plane in Figure~\ref{fig:toy_example} shows the distribution of the positive and negative samples in the training and testing datasets. A linear classifier may learn the acceptance and rejection regions split via the decision boundary shown in the figure. This decision boundary divides the two dimensional feature space in half. Even though there is a small overlap between the positive and negative classes, when evaluated against the negative and positive samples from the dataset there would be an acceptably low false positive rate. However if we construct a vector by uniformly sampling the two features from the interval $[0, 1]$, with probability $1/2$ it will be an accepting sample. If this model could be queried through an API, an attacker is expected to find an accepting sample in two attempts. Arguably, such a system is insecure. 

\begin{figure}[th!]
	\centering
    \includegraphics[width=0.70\columnwidth]{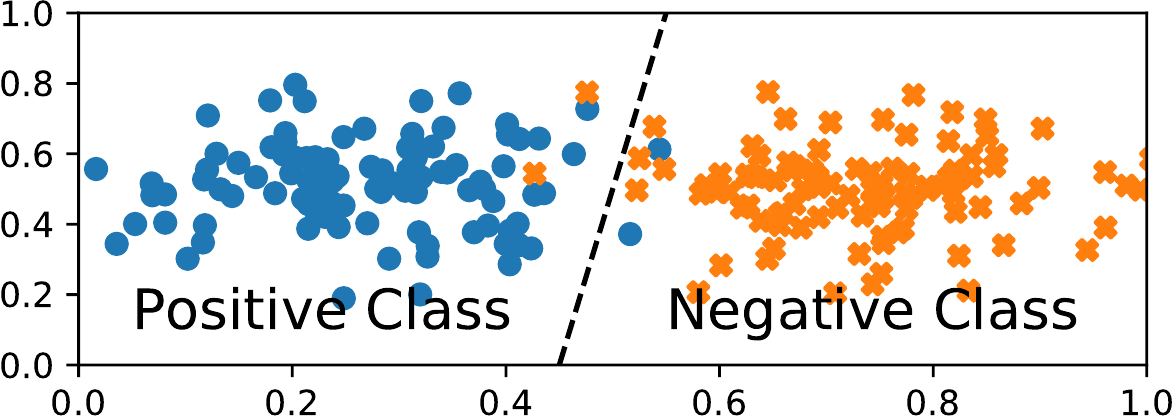}
    \vspace{-0.3cm}
  	\caption{Example feature space separation by a linear boundary between two classes. This demonstrates low FPR and FRR of test sample classification, yet allows approximately 50\% of the feature space to be accepted as positive. }
 	\label{fig:toy_example}
\end{figure}

Figure~\ref{fig:toy_example} illustrates that the acceptance region can be larger than the region covered by the target user's samples. However, in the same example, the area covered by the target user's samples is also quite high, e.g., the convex hull of the samples. As we discuss next, in higher dimensions, the area covered by the positive and negative examples is expected to be concentrated in an exponentially small region~\cite{data-science}. However, the acceptance region does not necessarily follow the same trend. 


\subsection{Acceptance Region}
\label{sec:attack_area}
\descr{Notations.} 
Let $\mathbb{I} \coloneqq [0, 1]$ denote the unit interval $[0, 1]$, and let $\mathbb{I}^n \coloneqq [0, 1]^n$ denote the $n$-dimensional unit cube with one vertex at the origin. The unit cube represents the feature space with each (min-max normalized) feature taking values in $\mathbb{I}$. Let $f$ denote a model, i.e., an output of a machine learning algorithm (classifier) trained on a dataset $D = \{ (\mathbf{x}_i, y_i) \}_{i \in [m]}$, where each $\mathbf{x}_i$ is a feature vector and $y_i \in \{+1, -1\}$ its label. The label $+1$ indicates the positive class (target user) and $-1$ the negative class (one or more other users of the authentication system). We may denote a positive example in $\mathbf{x} \in D$ by $\mathbf{x}^+$, and any negative example by $\mathbf{x}^-$. The model $f$ takes feature vectors $\mathbf{x} \in \mathbb{I}^n$ as input and outputs a \emph{predicted} label $\hat{y} \in \{+1, -1\}$. 

\descr{Definitions.} \emph{Acceptance region} of a model $f$ is defined as
\begin{equation}
    A_f \coloneqq \{ \mathbf{x} \in \mathbb{I}^n : f(\mathbf{x}) = +1 \},
\end{equation}
The $n$-dimensional volume of $A_f$ is denoted $\text{Vol}_n(A_f)$. The definition of acceptance region is analogous to the notion of decision regions in decision theory~\cite[\S 1.5]{bishop-pattern}. We will often misuse the word acceptance region to mean both the region or the volume covered by the region where there is no fear of ambiguity. Let FRR and FPR be evaluated on the training dataset $D$.\footnote{In practice, the FRR and FPR are evaluated against a subset of $D$ called a holdout or testing dataset.}
Let $\mathbf{x} \sample{\$} \mathbb{I}^n$ denote sampling a feature vector $\mathbf{x}$ uniformly at random from $\mathbb{I}^n$. In a \emph{random input attack}, the adversary samples $\mathbf{x} \sample{\$} \mathbb{I}^n$ and gives it as input to $f$. The attack is successful if $f(\mathbf{x}) = + 1$. The success probability of a random guess is defined as
\begin{equation}
\label{eq:ar}
\Pr  [ f(\mathbf{x}) = + 1 : \mathbf{x} \sample{\$} \mathbb{I}^n ].    
\end{equation}
Since the $n$-volume of the unit cube is 1, we immediately see that the above probability is exactly $\text{Vol}_n(A_f)$. Thus, we shall use the volume of the acceptance region as a direct measure of the success probability of random guess. Finally, we define the \emph{rejection region} as $\mathbb{I}^n - A_f$. It follows that the volume of the rejection region is $1 - \text{Vol}_n(A_f)$.

\descr{Existence Results.}
Our first observation is that even if the FPR of a model is zero, its acceptance region can still be non-zero. Note that this is not evident from the fact that there are positive examples in the training dataset $D$: the dataset is finite and there are infinite number of vectors in $\mathbb{I}^n$, and hence the probability of picking these finite positive examples is zero.
\begin{proposition}
There exists a dataset $D$ and a classifier with output $f$ such that $\text{FRR} = \text{FPR} = 0$, and $\text{Vol}_n(A_f) > 0$.
\end{proposition}
\begin{proof}
Assume a dataset $D$ that is linearly separable. This means that there exists a hyperplane denoted $H(\mathbf{x})$ such that for any positive example $\mathbf{x}^+ \in D$, we have $H(\mathbf{x}^+) > 0$ and for any negative example in $D$ we have $H(\mathbf{x}^-) < 0$. Consider the perceptron as an example of a classifier which constructs a linear model: $f_{\mathbf{w}, b} (\mathbf{x}) = + 1$ if $\langle \mathbf{w}, \mathbf{x} \rangle + b > 0$, and $-1$ otherwise. Since the data is linearly separable, the perceptron convergence theorem states that the perceptron learning algorithm will find a solution, i.e., a separating hyperplane~\cite{rosenblatt}. Intersecting this hyperplane $\langle \mathbf{w}, \mathbf{x} \rangle + b = 0$ with the unit cube creates two sectors. The sector where $f_{\mathbf{w}, b} (\mathbf{x}) = + 1$ is exactly the acceptance region $A_{f_{\mathbf{w}, b}}$. The $n$-volume of $A_{f_{\mathbf{w}, b}}$ cannot be zero, since otherwise it is one of the sides of the unit cube with dimension less than $n$, implying that all points $\langle \mathbf{w}, \mathbf{x} \rangle + b > 0$ lie outside the unit cube. A contradiction, since $\text{FRR} = 0$ (there is at least one positive example). 
\end{proof}

A non-zero acceptance region is not necessarily a threat. Of practical concern is a \emph{non-negligible} volume. Indeed, the volume may be negligible requiring a large number of queries to $f$ before an \emph{accepting sample} is produced. The following result shows that there are cases in which the acceptance region can be non-negligible.
\begin{proposition}
\label{prop:ar-half}
There exists a dataset $D$ and a classifier with output $f$ such that $\text{FRR} = \text{FPR} = 0$, and $\text{Vol}_n(A_f) \ge 1/2$.
\end{proposition}
\begin{proof}
Consider again the perceptron as an example of a classifier. Let $D$ be a dataset such that for all positive examples $\mathbf{x}^+$, we have $x^+_1 > 0.5$, and for all negative examples $x^-_1 < 0.5$. The rest of the features may take any value in $\mathbb{I}$. The resulting data is linearly separable by the ($n-1$ dimensional) hyperplane $x_1 - 0.5 = 0$. 
Initialize the perceptron learning algorithm with $w_1 = 1$, $w_i = 0$ for all $2 \le i \le n$, and $b = -0.5$. The algorithm then trivially stops with this hyperplane. Clearly, with this hyperplane, we have $\text{FRR} = 0$, $\text{FPR} = 0$, and the acceptance region is 1/2. 
\end{proof}

The above examples illustrate the high probability of success of the random input attack due to non-negligible acceptance region. However, the example used is contrived. In practice, datasets with a  ``nice'' distribution as above are seldom 
encountered, and the model is more likely to exhibit non-zero 
generalization error (a tradeoff between FRR and FPR). Also, in practice, more sophisticated classifiers such as the support 
vector machine or deep neural networks are used instead of the perceptron. However, we shall 
demonstrate that the issue persists in real datasets and 
classifiers used in practice. We remark that we are interested in the case when $\text{Vol}_n(A_f) > \text{FPR}$, since arguably it is misleading to use the FPR as a measure of security of such an authentication system. This could happen even when the FPR is non-zero. When and why would this case occur? We explain this in the following.

\descr{Real Data and High Dimensions.} 
We first discretize the feature space. 
For a given positive integer $B$, let $\mathbb{I}_B$ denote the binned (or discrete) version of the interval $\mathbb{I}$ partitioned into $B$ equally sized bins. Clearly, each bin is of width $1/B$. Let $\mathbb{I}_B^n$ denote the discretized feature space. Given a set of feature values from $\mathbb{I}$, we say that a bin in $\mathbb{I}_B$ is \emph{filled} if there are $> \epsilon_n$ feature values falling in that bin, where $\epsilon_n$ is a cutoff to filter outliers. The number of filled bins is denoted by $\alpha$. Clearly $\alpha \le B$. See Figure~\ref{fig:bins}. 


\begin{figure}[th!]
	\centering
	\vspace{-1mm}
\resizebox{0.55\columnwidth}{!}{%
\begingroup%
  \makeatletter%
  \providecommand\rotatebox[2]{#2}%
  \ifx\svgwidth\undefined%
    \setlength{\unitlength}{203.66368408bp}%
    \ifx\svgscale\undefined%
      \relax%
    \else%
      \setlength{\unitlength}{\unitlength * \real{\svgscale}}%
    \fi%
  \else%
    \setlength{\unitlength}{\svgwidth}%
  \fi%
  \global\let\svgwidth\undefined%
  \global\let\svgscale\undefined%
  \makeatother%
  \begin{picture}(1,0.34838103)%
    \put(0,0){\includegraphics[width=\unitlength,page=1]{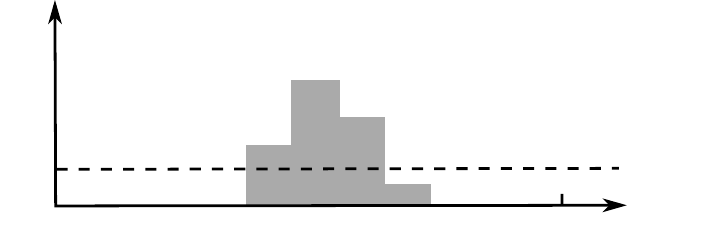}}%
    \put(0.04596194,0.01451667){\color[rgb]{0,0,0}\makebox(0,0)[lb]{\smash{\textbf{0}}}}%
    \put(0.78498776,0.01154079){\color[rgb]{0,0,0}\makebox(0,0)[lb]{\smash{\textbf{1}}}}%
    \put(0.88676664,0.09652185){\color[rgb]{0,0,0}\makebox(0,0)[lb]{\smash{\textbf{$\epsilon_n = 2$}}}}%
    \put(0.03075977,0.07302251){\color[rgb]{0,0,0}\rotatebox{90}{\makebox(0,0)[lb]{\smash{\textbf{Frequency}}}}}%
    \put(0.38841542,0.00059937){\color[rgb]{0,0,0}\makebox(0,0)[lb]{\smash{\textbf{Bins}}}}%
  \end{picture}%
\endgroup%
}
\vspace{-3.5mm}
\caption{The binned version $\mathbb{I}_B$ of the unit interval $\mathbb{I}$. Each bin is of width $1/B$ ($B$ not specified). The number of filled bins is $\alpha = 3$, with a cut-off of $\epsilon_n = 2$.}
\label{fig:bins}

\end{figure}
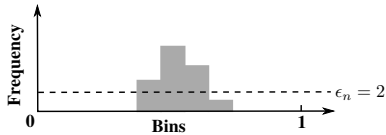

For the $i$th feature, let $\alpha_i^{+}$ denote the number of bins filled by all positive examples in $D$. We define:
\[
R^+ \coloneqq \frac{1}{B^n}\prod_{i = 1}^n \alpha^+_i 
\]
as the \emph{volume of the true positive region}. We define $\alpha_i^{-}$ and $R^-$ analogously as the volume of the true negative region. Let $c \in [0, 1]$ be a constant. If each of the $\alpha_i^{+}$'s is at most $cB$, then we see that $R^+ = c^{-n}$. For instance, if $c \le 1/2$, then $R^+ \le 2^{-n}$. In other words, the volume of the region spanned by the user's own samples is exponentially small as compared to the volume of the unit cube. In practice, the user's data is expected to be normally distributed across each feature, implying that the $\alpha_i^{+}$'s are much smaller than $B/2$, which makes the above volume a loose upper bound. The same is true of the $\alpha_i^{-}$'s. Figure~\ref{fig:alpha-hist} shows the filled bins from one of the features in the Face Dataset (see Section~\ref{sub:datasets}). For the same dataset, the average true positive region is $5.781 \times 10^{-98}$ (with a standard deviation of $\pm 2.074 \times 10^{-96}$) and the average true negative region is $1.302 \times 10^{-55}$ (with a standard deviation of $\pm 2.172 \times 10^{-54}$) computed over 10,000 iterations considering a 80\% sample of the target user's data, and a balanced sample of other users.\footnote{We compute the true positive and negative region by only considering the minimum and maximum feature values covered by each user for each feature with binning equal to the floating point precision of the system. Thus, this is a conservative estimate of the true positive region.}



We thus expect a random vector from $\mathbb{I}^n$ to be outside the region spanned by the target user with overwhelming probability. Thus, if a classifier defines an acceptance region tightly surrounding the target user's data, the volume of the acceptance region will be negligible, and hence the random input attack will not be a threat. However, as we shall show in the next sections, this is not the case in practice. 

\descr{Factors Effecting Acceptance Region.} We list a few factors which effect the volume of the acceptance region.

\begin{itemize}[leftmargin=5mm]
    \item One reason for a high acceptance region is that the classifier is not penalized for classifying \emph{empty space} in the feature space as either positive or negative. For instance, consider Figure~\ref{fig:alpha-hist}. There is significant empty space for the feature depicted in the figure:  none of the positive or negative samples have the projected feature value in this space. A classifier is generally trained with an objective to minimize the misclassification rate or a loss function (where, for instance, there is an asymmetrical penalty between true positives and false positives)~\cite{bishop-pattern}. These functions take input from the dataset $D$. Thus, empty regions in the feature space which do not have examples in $D$ can be classified as either of the two classes without being penalized during training, resulting in a non-negligible acceptance region.
    \item The acceptance region is also expected to be big if there is high variance in the feature values taken by the positive examples. In this case, the $\alpha^{+}_i$'s will be much closer to $B$, resulting in a non-negligible volume $R^+$.
    \item On the other hand, the acceptance region is likely to be small if the variances of the feature values in the negative examples are high. The classifier, in order to minimize the FPR, will then increase the region where samples are rejected, which would in turn make the acceptance region closer in volume to the true positive region. 
\end{itemize}

\begin{figure}[th!]
	\centering
    \includegraphics[width=0.75\columnwidth]{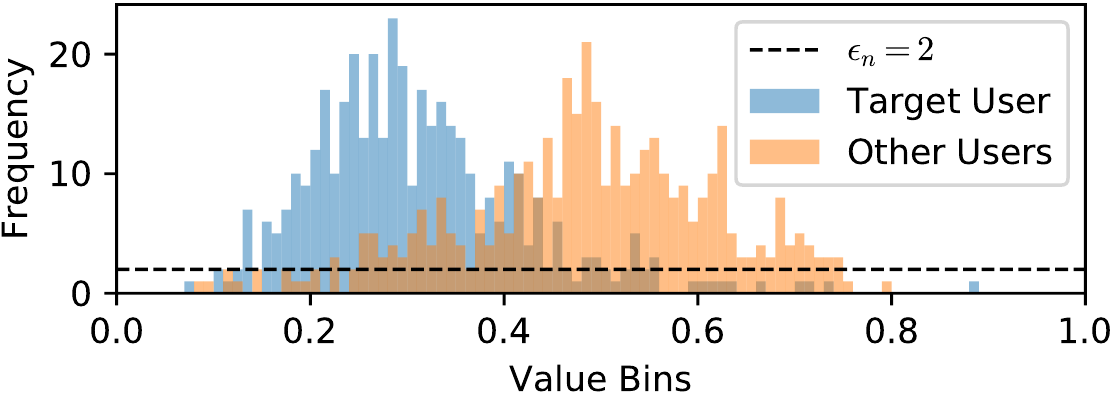}
    \vspace{-3mm}
  	\caption{The histogram of feature values of one of the features in the Face Dataset (cf. \S~\ref{sub:datasets}). Here we have $B = 100$. The number of filled bins for the target user is $\alpha_i^+ = 35$ (with 400 samples), and for the negative class (10 users; same number of total samples) it is $\alpha_i^- = 50$. A total of $24$ bins are not filled by any of the two classes, implying that (approximately) $0.24$ of the region for this feature is empty.}
 	\label{fig:alpha-hist}
 	\vspace{-4mm}
\end{figure}


We empirically verify these observations in Section~\ref{sec:synthetic}.  The last observation also hints at a possible method to tighten the acceptance region around the region spanned by the target user: generate random noise around the target user's vectors and treat it as belonging to the negative class. We demonstrate the effectiveness of this method in Section~\ref{sec:mitigation}. Jumping ahead, if the noise is generated from an appropriate distribution, this will have minimal impact on the FRR and FPR of the model.




%
\section{Evaluation on Biometric Systems}
\label{sec:realscheme}

To evaluate the issue of acceptance region on real-world biometric systems, we chose four different modalities: gait, touch, face and voice. The last two modalities are used as examples of user authentication at the point of entry into a secured system, whilst gait and touch are often used in continuous authentication systems~\cite{barbello2016continuous}. We first describe the four biometric datasets, followed by our evaluation methodology, the machine learning algorithms used, and finally our results and observations.

\subsection{The Biometric Datasets}
\label{sub:datasets}

\subsubsection{Activity Type (Gait) Dataset}
The activity type dataset \cite{anguita2013public}, which we will refer to as the ``gait'' dataset, was collected for human activity recognition. Specifically its aim is to provide a dataset for determining if a user is sitting, laying down, walking, running, walking upstairs or downstairs, etc. However, as the dataset retains the unique identifiers for users per biometric record, we re-purpose the dataset for authentication. This dataset contains 30 users, with an average of $343 \pm 35$ ($\text{mean} \pm \text{SD}$) biometric samples per user, there is an equal number of activity type samples for each user. For the purpose of authentication, we do not isolate a specific type of activity. Instead, we include them as values of an additional feature. The activity type feature increases the total number of features to 562. We will refer to these features as \emph{engineered} features as they are manually defined (e.g., by an expert) as opposed to \emph{latent} features extracted from a pre-trained neural network for the face and voice datasets.


\subsubsection{Touch Dataset}
The UMDAA-02 Touch Dataset \cite{mahbub2016active} is a challenge dataset to provide data for researchers to perform baseline evaluations of new touch-based authentication systems. Data was collected from 35 users, with an average of $3667 \pm 3012$ 
swipes per user. This dataset was collected by lending mobile devices to the participants over a prolonged period of time. The uncontrolled nature of the collection produces a dataset that accurately reflects swipe interactions with constant and regular use of the device. This dataset contains every touch interaction performed by the user including taps. In a pre-processing step we only consider sequences with more than 5 data points as swipes. Additionally, we set four binary features to indicate the direction of the swipe, determined from the dominant vertical and horizontal displacement. We retained all other features in \cite{mahbub2016active} bar inter-stroke time, as we wished to treat each swipe independently, without chronological order. We substitute this feature with half-time of the stroke. This produces a total of 27 engineered touch features. 

\subsubsection{Face Dataset}
\label{subsub:face-dataset}
FaceNet~\cite{schroff2015facenet} proposes a system based on neural networks that can effectively learn embeddings (feature vectors) that represent uniquely identifiable facial information from images. Unlike engineered features, these embeddings may not be directly explainable as they are automatically extracted by the underlying neural network. This neural network can be trained from any dataset containing labeled faces of individuals. There are many sources from which we can obtain face datasets, CASIA-WebFace~\cite{yi2014learning}, VGGFace2~\cite{cao2018vggface2} and Labeled Faces in the Wild (LFW)~\cite{LFWTechUpdate} are examples of such datasets. However, with a pre-trained model, we can conserve the time and resources required to re-train the network. The source code for FaceNet~\cite{facenet-source}
contains two pretrained models available for public use (at the time of writing): one trained on CASIA-WebFace, and another trained on VGGFace2. We opt to use a model pre-trained on VGGFace2\footnote{(20180402-114759) is the identifier of pre-trained model used.} 
, while retaining CASIA-WebFace as our dataset for classifier training. We choose to use different datasets for the training of the embeddings and the classifiers to simulate the exposure of the model to never before seen data. Our face dataset is a subset of CASIA-WebFace containing only the top 100 identities with the largest number of face images (producing $447 \pm 103$ images per individual). This model produces 512 latent features from input images of pixel size 160x160 which have been centered and aligned. Recall that face alignment involves finding a bounding box on the face on an image, before cropping and resizing to the requested dimensions.



\subsubsection{Speaker Verification (Utterances)}%
VoxCeleb~\cite{nagrani2017voxceleb}, and VoxCeleb2~\cite{chung2018voxceleb2} are corpuses of spoken recordings by celebrities in online media. These recordings are text-independent, i.e., the phrase uttered by the user is not pre-determined. Text-independent speaker verification schemes depart from text-dependent verification schemes in which the individual is bound to repeat a pre-determined speech content. Thus, the task of text-independent verification (or identification) is to distinguish how the user speaks as an individual, instead of how the user utters a specific phrase. The former objective is an arguably harder task. Despite the increased difficulty, researchers have trained neural networks to convert speaker utterances into a set of latent features representing how individuals speak. These works have also released their models to the public, increasing the accessibility of speaker verification to developers. We opt to use the pre-trained model of VoxCeleb~\cite{nagrani2017voxceleb}, with utterances from VoxCeleb2~\cite{chung2018voxceleb2}. From VoxCeleb2, we only use the test portion of the dataset, which contains 118 Users with an average of $406 \pm 87$ utterances. VoxCeleb was trained as a Siamese neural network~\cite{siamese-bromley} for one-shot comparison between two audio samples. A Siamese network consists of two identical branches that produce two equal size outputs from two independent inputs for distance comparison. To fit the pre-trained model into our evaluation of ML-based models, we extract embeddings from one of the twin networks and disregard the second branch. The 1024-length embedding is then used as the feature vector within our evaluation.

\subsection{Evaluation Methodology}



\label{sub:eval-method}
In our creation of biometric models for each user, we seek to obtain the baseline performance of the model with respect to the ability of negative user samples gaining access (i.e. FPR), and the measured Acceptance Region (AR). We use the following methodology to evaluate these metrics for each dataset and each classification algorithm.
\begin{enumerate}[leftmargin=5mm]
    \item We min-max normalize each extracted feature over the entire dataset between 0 and 1.
    \item We partition the dataset into a $(70\%, 30\%)$ split for training and testing sets, respectively.
    \item For both training and testing samples, we further sample an equal number of negative samples from every other user such that the total number of negative samples are approximately equal to the number of samples from the target user, representing the positive class, i.e., the positive and negative classes are balanced.
    \item Using the balanced training set from step 3, we train a two-class classifier defining the target user set as the positive class, and all remaining users as negative.
    \item We test the trained model using the balanced testing set from step 3. This establishes the FRR and FPR of the system.
    \item We uniformly sample one million vectors from $\mathbb{I}_n$, where $n$ is the dimension of the extracted features. Testing the set of vectors against the model measures the acceptance region (AR).
    \item We record the confidence values of the test prediction for the user's positive test samples, other users' negative test samples, and the uniformly sampled vectors. These confidence values produce ROC curves for FRR, FPR and AR.
    \item Repeat steps 3-7 by iterating through every user in the dataset as the target user.
\end{enumerate}

\begin{remark}
\label{rem:monte-carlo}
In general, the decision regions (accept and reject in the case of authentication) learned by the classifiers can be quite complex~\cite{region-dnn}. Hence, it is difficult to determine them analytically, despite the availability of learned model parameters. We instead use a Monte Carlo method by sampling random feature vectors from $\mathbb{I}_n$ where each feature value is sampled uniformly at random from $\mathbb{I}$. With enough samples (one million used in our experiments, \blue{ and averaged over 50 repetitions}), the fraction of random samples accepted by the classifier serves as an estimate of the acceptance region as defined by Eq.~\ref{eq:ar} due to the law of large numbers.
\end{remark}

\blue{
\begin{remark}
Our evaluation of the biometric systems is using the mock attacker model (samples from the negative class modelled as belonging to an attacker) as it is commonly used~\cite{eberz2017evaluating}. We acknowledge that there are other attack models such as excluding the data of the attacker from the training set~\cite{eberz2017evaluating}. Having the attacker data included in the training dataset, as in the mock attacker model, yields better EER. On the other hand, it is also likely to lower the AR of the system, due to increased variance in the negative training dataset. Thus, the use of this model does not inflate our results.
\end{remark}
}

\blue{
\begin{remark}
\label{rem:balanced}
We have used balanced datasets in our experiments, i.e., the number of positive and negative samples being the same. It is true that a balanced dataset is not ideal for minimizing AR; more negative samples may reduce the acceptance region. However, an unbalanced dataset, e.g., more negative samples than positive samples, may be biased towards the negative class, resulting in misleadingly high accuracy~\cite{eberz2017evaluating, sugrim2019robust}. A balanced dataset yields the best EER without being biased towards the positive or negative class.
\end{remark}
}

\subsection{Machine Learning Classifiers}
Our initial hypothesis (Section~\ref{sec:attack_def}) stipulates that AR is related to the training data distribution, and not necessarily to any weakness of the classifiers learning from the data. To demonstrate this distinction, we elected four different machine learning algorithms: Support Vector Machines (SVM) with a linear kernel (LinSVM), SVM with a radial basis function kernel (RBFSVM), Random Forests (RNDF) and Deep Neural Networks (DNN). Briefly, SVM uses the training data to construct a decision boundary that maximizes the distance between the closest points of different classes (known as support vectors). The shape of this boundary is dictated by the kernel used; we test both a linear and a radial kernel. 
Random Forests is an aggregation of multiple decision tree learners formally known as an ensemble method. Multiple learners in the aggregation are created through bagging, whereby the training dataset is split into multiple subsets, each subset training a distinct decision tree. The decisions from the multiple models are then aggregated to produce the random forest's final decision. 
DNNs are a class of machine learning models that contain hidden layers between an input and an output layer; each layer containing neurons that activate as a function of previous layers. Specifically we implement a convolutional neural network with hidden layers leading to a final layer of our two classes, accept and reject.
All four of these machine learning models are trained as supervised learners. As such, we provide the ground truth labels to the model during training. 

The linear SVM was trained with $C=10^4$, and default values included within Scikit-learn's Python library for the remaining parameters \cite{scikit-learn}. For radial SVM we also used $C=10^4$ while keeping the remaining parameters as default. The Random Forests classifier was configured with 100 estimators. DNNs were trained with TensorFlow Estimators~\cite{tensorflow-estimators}
with a varying number of internal layers depending on the dataset. The exact configurations are noted in Appendix~\ref{sec:dnn_config}.

\begin{remark}
We reiterate that our trained models are reconstructions of past works.
However, we endeavor that our models recreate error rates similar to the originally reported values on the same dataset. On Mahbub et al.'s touch dataset~\cite{mahbub2016active}, the authors achieved 0.22 EER with a RNDF classifier, by averaging 16 swipes for a single authentication session. We are able to achieve a comparable EER of 0.21 on RNDF without averaging. 
For face authentication, we evaluate a subset of CASIA-Webface, consequently there is no direct comparison. The original Facenet accuracy in verifying pairs of LFW~\cite{LFWTechUpdate} faces is 98.87\%~\cite{schroff2015facenet}, but our adoption of model-based authentication is closer to~\cite{explainable-AI}, unfortunately the authors have fixed a threshold for 0 FPR without reporting their TPR.
Nagrani, Chung and Zisserman's voice authenticator~\cite{nagrani2017voxceleb} reports an EER of 0.078 on a neural network. Our classifiers achieve EERs of 0.03, 0.02, 0.04 and 0.12, which are within range of this benchmark. Our gait authenticator is the exception, it has not been evaluated for authentication with it's mixture of activity types. However, a review of gait authentication schemes can be found at~\cite{sprager2015inertial}.


\end{remark}

\subsection{Acceptance Region: Feature Vector API}



In this section, we evaluate the acceptance region (AR) by comparing it against FPR for all 16 authentication configurations (four datasets and four classifiers). In particular, we display ROC curves showing the trade-off between FPR and FRR against the acceptance region (AR) curve as the model thresholds are varied. These results are averaged over all users. While this gives an average picture of the disparity between AR and FPR, it does not highlight that for some users AR may be substantially higher than FPR, and vice versa. In such a case, the average AR might be misleading. Thus, we also show scattered plots showing per-user AR and FPR, where the FPR is evaluated at EER. The per-user results have been averaged over 50 repetitions to remove any bias resulting from the sampled/generated vectors. The individual user AR versus FPR scatter plots are shown in Figure~\ref{fig:all_ind}, and the (average) AR curves against the ROC curves are shown in Figure~\ref{fig:all_roc}. 

%

\begin{remark}
EER is computed in a best effort manner, with only 100 discretized threshold values, to mitigate the storage demands of the 1M uniformly random vectors measuring AR. Unfortunately, there are some instances whereby the FRR and FPR do not match exactly, as the threshold step induces a large change in both FRR and FPR. Only 1/16 classifiers exhibit an FPR-FRR discrepancy greater than 1\%.
\end{remark}


\begin{figure}[t]
    \centering
  \subfloat[Gait\label{fig:gait_ind}]{%
       \includegraphics[width=0.46\linewidth]{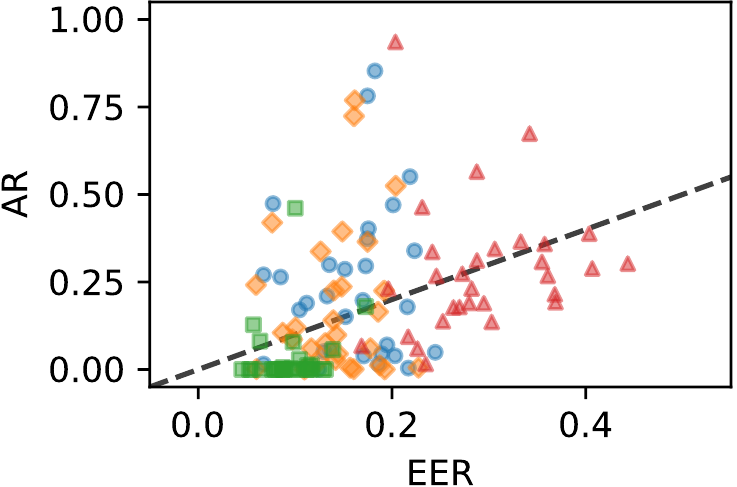}}
    \hfill
  \subfloat[Touch\label{fig:touch_ind}]{%
        \includegraphics[width=0.46\linewidth]{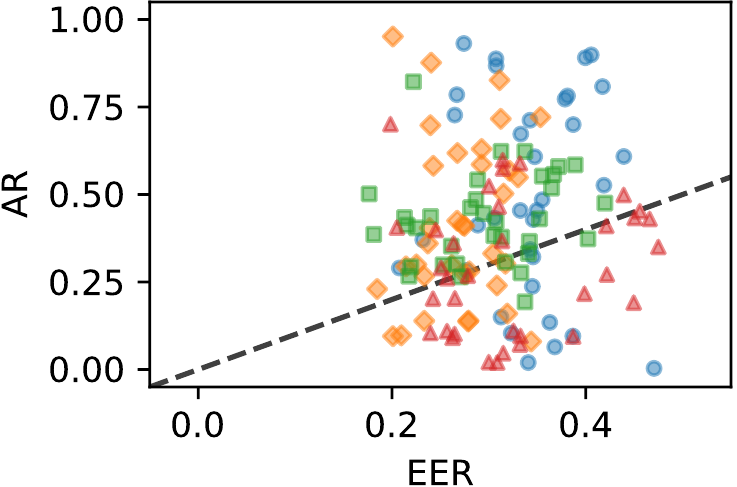}}
    \vspace{-0.3cm}
    \\
  \subfloat[Face\label{fig:face_ind}]{%
        \includegraphics[width=0.46\linewidth]{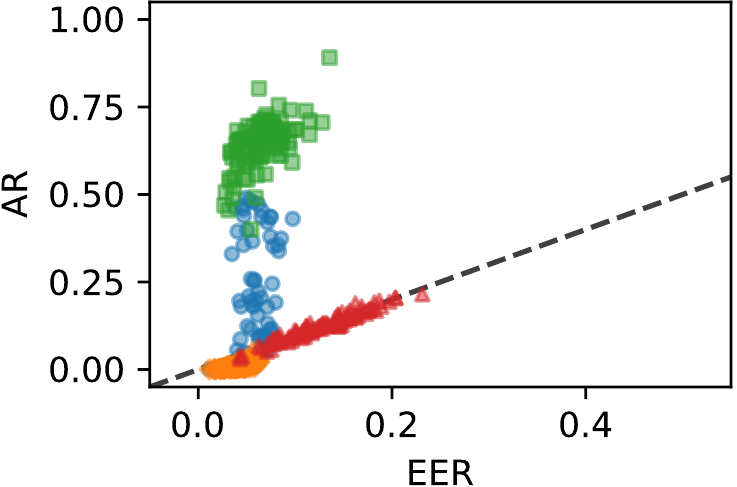}}
    \hfill
  \subfloat[Voice\label{fig:voice_ind}]{%
        \includegraphics[width=0.46\linewidth]{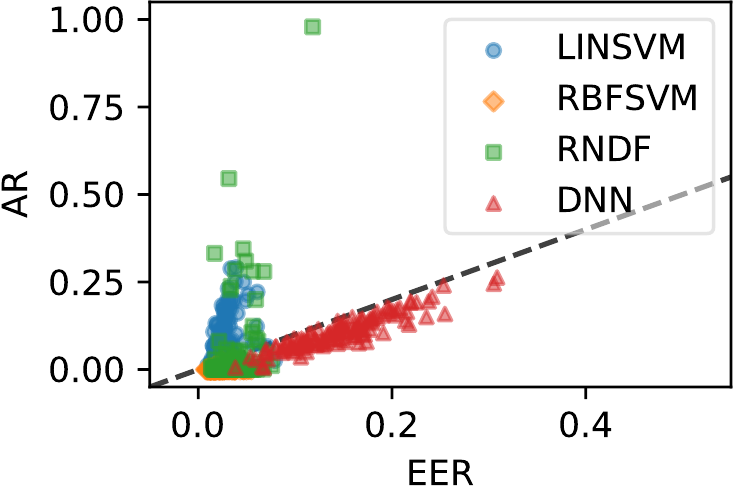}}
  \vspace{-0.2cm}
  \\
  \caption{Individual user scatter of AR and FPR. In a majority of configurations, there is no clear relationship between AR and FPR, with the exception of the RBFSVM and DNN classifiers for face and voice authentication. 
  \vspace{-1mm}
  }
  
  \label{fig:all_ind} 
\end{figure}

\begin{figure*}
  \centering
  \subfloat[Gait Average ROC\label{fig:gait_roc}]{%
       \includegraphics[width=1.0\linewidth]{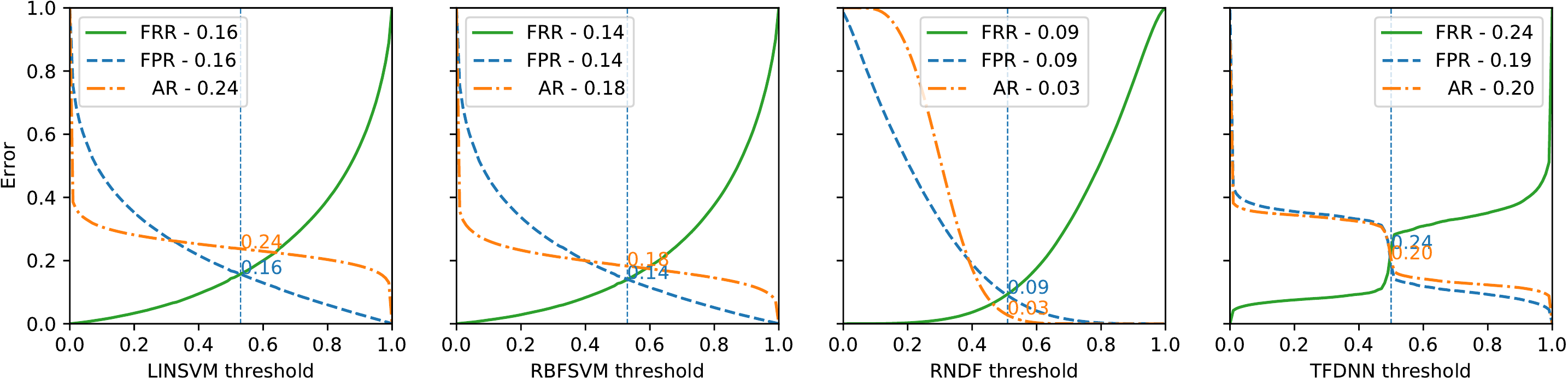}}\\
  \subfloat[Touch Average ROC\label{fig:touch_roc}]{%
        \includegraphics[width=1.0\linewidth]{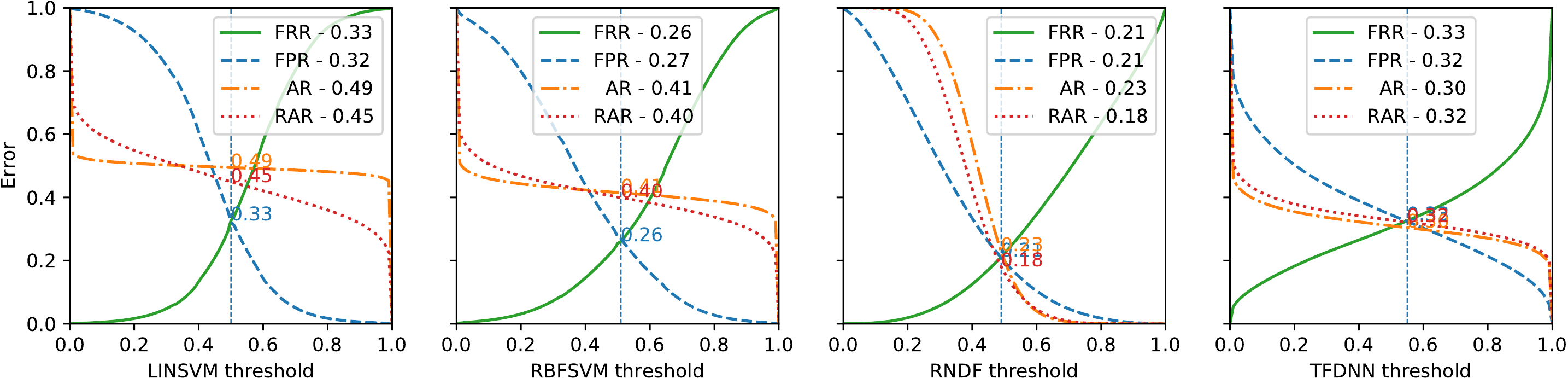}}\\
  \subfloat[Face Average ROC\label{fig:face_roc}]{%
        \includegraphics[width=1.0\linewidth]{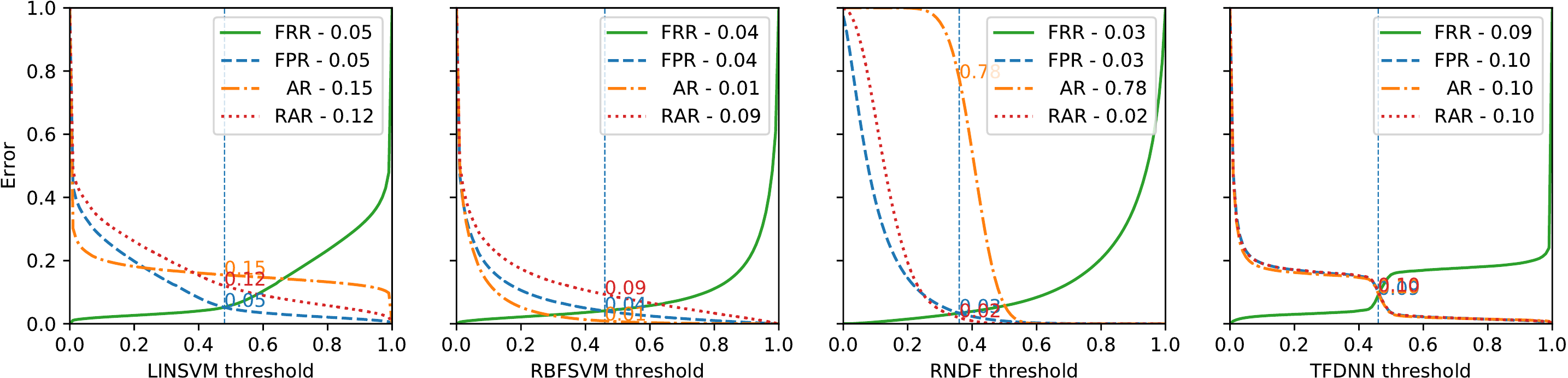}}\\
  \subfloat[Voice Average ROC\label{fig:voice_roc}]{%
        \includegraphics[width=1.0\linewidth]{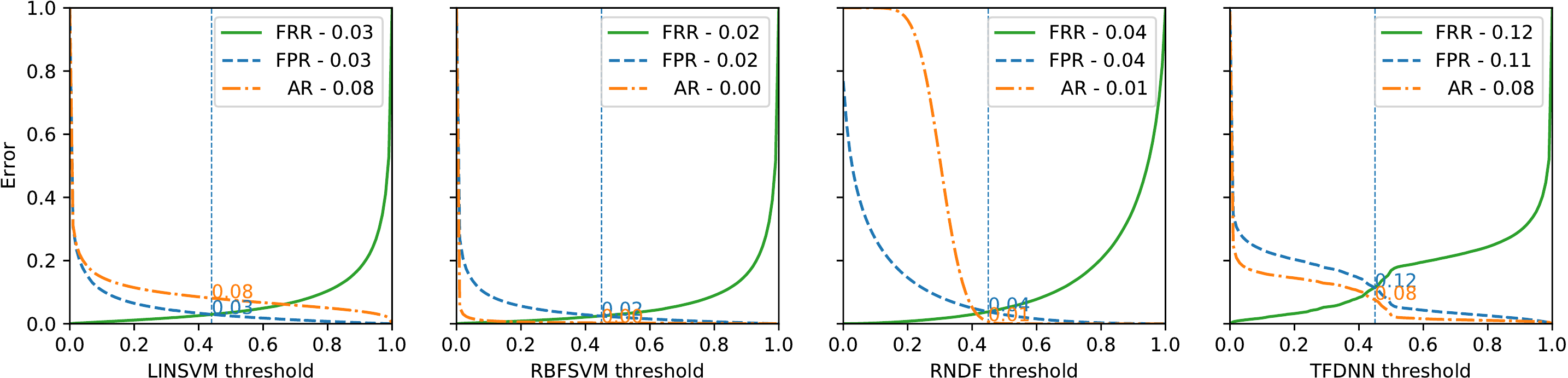}}
  \caption{ROC curve versus the AR and RAR curves for all configurations. The EER is shown as a dotted vertical blue line. The FRR, FPR, AR and RAR values shown in the legend are evaluated at EER (FPR = FRR). The RAR is only evaluated on the Touch and Face datasets.} 
  \label{fig:all_roc} 
\end{figure*}

\subsubsection{Gait Authentication}
\label{sec:classify_gait}
%
%
%
Figure~\ref{fig:gait_ind} shows AR against FPR of every user in the activity type (gait) dataset. Recall that in this figure FPR is evaluated at EER. The dotted straight line is the line where AR equals FPR (or ERR). We note that there is a significant proportion of users for which AR is greater than FPR, even when the latter is reasonably low. For instance, in some cases AR is close to 1.0 when the FPR is around 0.2. Thus, a random input attack on systems trained for these target users will be successful at a rate significantly higher than what is suggested by FPR. We also note that the two SVM classifiers have higher instances of users for whom AR surpasses FPR. Figure~\ref{fig:gait_roc} shows the AR curve averaged across all users against the FPR and FRR curves for all four classifiers. We can see that AR is higher than the ERR (represented by the dotted vertical line) for the two SVM classifiers. For the remaining two classifiers, AR is lower than EER. However, by looking at the AR curve for RNDF, we see that the AR curve is well above the FPR curve when $\text{FRR} \le 0.3$. This can be specially problematic if the threshold is set so as to minimize false rejection at the expense of false positives. We also note that the AR curve for DNN closely follows the FPR curve, which may suggest that the AR is not as problematic for this classifier. However, by looking at Figure~\ref{fig:gait_ind}, we see that this is misleading since for some users the AR is significantly higher than FPR, making them vulnerable to random input attacks. Also, note that the AR generally decreases as the threshold is changed at the expense of FRR. However, except for RNDF, the AR for the other three classifiers is significantly higher than zero even for FRR values close to 1.

\subsubsection{Touch (Swipe) Authentication}
\label{sec:classify_touch}

The touch authenticator has the highest EER of all four biometric modalities. Very few users attained an EER lower than 0.2 as seen in Figure~\ref{fig:touch_ind}. This is mainly because we consider the setting where the classification decision is being made after each input sample. Previous work has shown EER to improve if the decision is made on an average vector of a few samples some work~\cite{unobservable, touchalytics, chauhan2017behaviocog}. Nevertheless, since our focus is on AR, we stick to the per-sample decision setting. Figure~\ref{fig:touch_ind} shows that more than half of the users have ARs larger than FPR, and in some cases where the FPR is fairly low (say 0.2), the AR is higher than 0.5. Unlike gait authentication where RNDF classifier had ARs less than FPR for the majority of the users, all four algorithms for touch authentication display high vulnerability to the AR based random input attack. When viewing average results in Figure~\ref{fig:touch_roc}, we observe the average AR curve to be very `flat' for both SVM classifiers and DNN. This indicates that AR for these classifier remains mostly unchanged even if the threshold is moved closer to the extremes. RNDF once again is the exception, with the AR curve approaching 0 as the threshold is increased.




\subsubsection{Face Authentication}
\label{sec:classify_face}
Figure~\ref{fig:face_ind} shows that AR is either lower or comparable to FPR for RBFSVM and DNN. Thus, the FPR serves as a good measure of AR in these systems. However, AR for most users is significantly higher than FPR for LinSVM and RNDF. This is true even though the EER of these systems is comparable to the other two as seen in Figure~\ref{fig:face_roc}. For LINSVM, we have an average AR of 0.15 compared to an EER of 0.05. For RNDF, the situation is worse with the AR reaching 0.78 against an EER of 0.03. We also note that while the AR is equal to FPR for DNN, its value of 0.10 is still worrisome to be resistant to random input attack. The relatively high FPR for DNN as compared to RBFSVM is likely due to a limited set of training data available in training the neural network.

\subsubsection{Voice Authentication}
\label{sec:classify_speech}
Figure~\ref{fig:voice_ind} shows that once again LinSVM and RNDF have a significant proportion of users with AR higher than FPR, whereas for both RBFSVM and DNN the AR of users is comparable to FPR. Looking at the average ARs in Figure~\ref{fig:voice_roc}, we see that interestingly RNDF exhibits an average AR of 0.01 well below the ERR of 0.04. The average suppresses the fact that there is one user in the system with an AR close to 1.0 even with an EER of approximately 0.1, and two other users with an AR of 0.5 and 0.3 for which the EER is significantly below 0.1. Thus these specific users are more susceptible to the random input attack. Only LinSVM has an average AR (0.08) higher than EER (0.03). The average AR of DNN is lower than EER (0.11), but it is still significantly high (0.08). For RBFSVM we have an average AR close to 0.  



\subsection*{Observations}
In almost every configuration, we can observe that the average AR is either higher than the FPR or at best comparable to it. Furthermore, for some users the AR is higher than FPR even though the average over all users may not reflect this trend. This demonstrates that an attacker with no prior knowledge of the system can launch an attack against it via the feature vector API. Moreover, for both the linear and radial SVM kernels, and some instances of the DNN classifier, we observe a relatively flat AR curve as the threshold is varied, unlike the gradual convergence to 1 experienced by the FPR and FRR curves. 
These classifiers thus have a substantial acceptance region that accept samples as positives irrespective of the threshold. Random Forests is the only classifier where the AR curve shows significant drop as the threshold is varied. 
Random forests sub-divide the training dataset in a process called bagging, where each sub-division is used to train one tree within the forest. With different subsets of data, different training data points will be closer to different empty regions in feature space, thus producing varied predictions. Because the prediction confidence of RNDF is computed from the proportion of trees agreeing with a prediction, the lack of consensus within the ensemble of trees for the empty space may be the reason for the non-flat AR curve.


\subsection{Acceptance Rate: Raw Input API}
\label{sub:rar-results}

The results from the feature vector API are not necessarily reflective of the success rate of a random input attack via the raw input API. One reason for this is that the feature vectors extracted from raw inputs may or may not span the entire feature space, and as a consequence the entire acceptance region. For this reason, we use the term \emph{raw acceptance rate} (RAR) to evaluate the probability of successfully finding an accepting sample via raw random inputs. To evaluate RAR, we select the touch and face biometric datasets. The raw input of the touch authenticator is a time-series, whereas for the face authentication system it is an image. 




\subsubsection{Raw Touch Inputs}
We used a continuous auto-regressive process (CAR)~\cite{timesynth-source}
to generate random timeseries. We opted for CAR due to the extremely high likelihood of time-series values having a dependence on previous values. This time-series was then min-max scaled to approximate sensor bounds. For example the \texttt{x}-position has a maximum and minimum value of 1980 and 0 respectively, as dictated by the number of pixels on a smartphone screen. Both the duration and  length of the time-series were randomly sampled from reasonable bounds: 0.5 to 2.0 seconds and 30 to 200 data-points, respectively. The time-series was subsequently parsed by the same feature extraction process as a legitimate time-series, and the outputs scaled on a feature min-max scale previously fit on real user data. In total we generate 100,000 time-series\blue{, which are used to measure RAR over 50 repetitions of the experiment.}

The results of our experiments are shown in  Figure~\ref{fig:touch_roc}, with the curve labeled RAR showing the raw acceptance rate as the threshold of each of the classifiers is changed. As we can see, the RAR is large and comparable to AR. This seems to indicate that the region spanned by random inputs covers the acceptance region. However, on closer examination, this happens to be false. The average volume covered by the true positive region for the touch dataset (cf. Section~\ref{sec:attack_def}) is less than $1.289 \times 10^{-4} \pm 5.462 \times 10^{-4}$, yet the volume occupied by the feature vectors extracted from raw inputs is less than $2.609 \times 10^{-6}$. This is significantly smaller than the AR for all four classifiers. We will return to this observation shortly.



\subsubsection{Raw Face Inputs}

We generated 100,000 images of size 160x160 pixels, with uniformly sampled RGB values. Feature embeddings were then extracted from the generated images with the pre-trained Facenet model (cf. Section~\ref{subsub:face-dataset}). This set of 100,000 raw input vectors, was parsed by a min-max scaler fitted to real user data. We did not align the noisy images, as there is no facial information within the image to align. Note that alignment is normally used in face authentication to detect facial boundaries within an image. \blue{Again, we aggregate results over 50 repetitions to remove any potential biases.}

The results from these raw inputs are shown in Figure~\ref{fig:face_roc}. We note that the RAR curve behaves much more similarly to the FPR curve, than what was previously observed for raw touch inputs. Also, in the particular example of RBFSVM, we obtain an RAR of 0.09 which is significantly higher than the AR (0.01) at the equal error rate. We again computed the true positive region and found that the average is $6.562 \times 10^{-94} \pm 6.521 \times 10^{-93}$. However, the volume covered by the raw inputs (after feature extraction) is only $4.670 \times 10^{-390}$, which is negligible compared to the ARs (0.15, 0.01, 0.78 and 0.10 for all four classifiers). Additional analysis shows that only one other user's feature space overlapped with the space of raw inputs, with an overlapped area of $8.317 \times 10^{-407}$, many orders of magnitude smaller than both the positive users and the raw feature space itself. \

\subsection*{Observations} 
The threat of a random input attack via raw random inputs is also high, and in some cases greater than the FPR. However, the region spanned by the feature vectors from these raw inputs is exponentially small and hence does not span the acceptance region. Furthermore, the region also does not coincide with any true positive region. This implies that raw inputs may result in high raw acceptance rate due to the fact that the training data does not have representative vectors in the region spanned by raw inputs. We shall return to this observation when we discuss mitigation strategies in Section~\ref{sec:mitigation}.


%
\section{Synthetic Dataset}
\label{sec:synthetic}

\begin{figure*}[th!]
	\centering
    \includegraphics[width=1.0\textwidth]{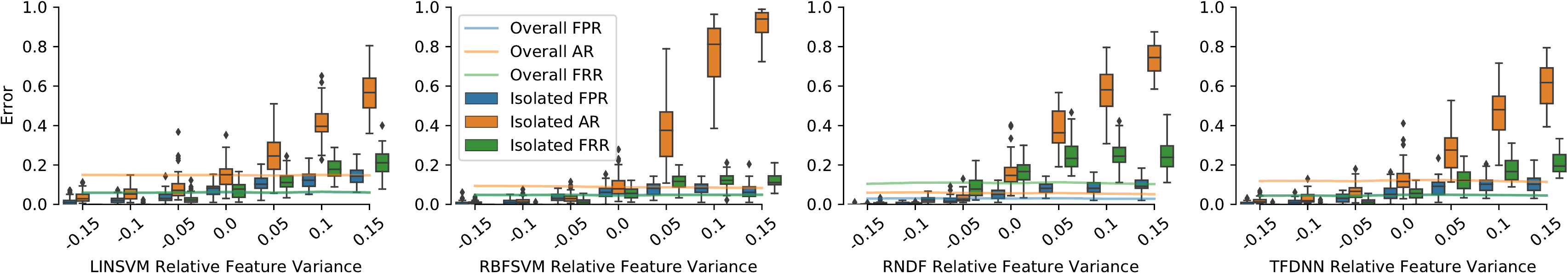}
    \vspace{-0.65cm}
  	\caption{A comparison between FPR, AR, four different ML architectures. Trained on synthetic data of 50 features of 50 user, of increasing variance within features for a singular user, repeated 50 times. Note how the system level AR and FPR remains unchanging, despite the isolated user's AR increasing substantially.}
 	\label{fig:isolate_userv}
 	\vspace{-0.10cm}
\end{figure*}

\begin{figure*}[th!]
	\centering
    \includegraphics[width=1.0\textwidth]{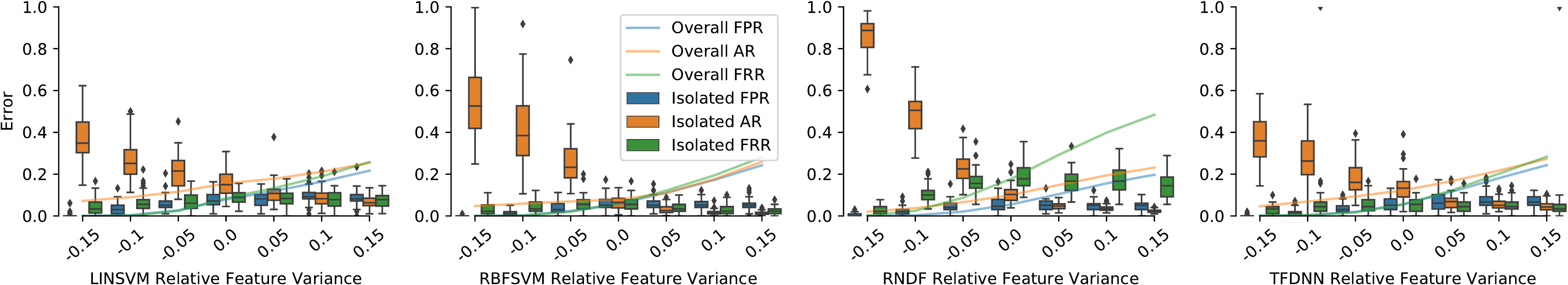}
    \vspace{-0.65cm}
  	\caption{A comparison between FPR, AR, four different ML architectures. Trained on synthetic data of 50 features of 50 user, of increasing variance within features of all other users except a singular user, repeated 50 times. The x-axis denotes the relative SD of the population compared with the isolated user.}
 	\label{fig:isolate_otherv}
 	\vspace{-0.05cm}
\end{figure*}

The analysis in the previous section was limited in the sense that we could not isolate the reasons behind the discrepancy between AR and FPR. Indeed, we saw that for some configurations (dataset-classifier pairs), the AR curve nicely followed the FPR curve, e.g., the face dataset and DNN (Figure~\ref{fig:face_ind}), where as for others this was not the case. In order to better understand the factors effecting AR, in this section we attempt to empirically verify the hypothesized factors effecting the acceptance region outlined in Section~\ref{sec:attack_def}. Namely, high feature variance in a target user's samples is likely to increase AR, and low feature variance in the user samples in the negative class is expected to result in high AR. In both these cases, we expect to achieve a reasonably low EER, but AR may still be significantly greater than FPR. Moreover, if these factors are indeed true, we expect to see similar behavior across all classifiers. To test this we create a synthetic model of a biometric dataset. 



\subsection{Simulating a Biometric Dataset}
Let $\mathcal{N}(\mu, \sigma^2)$, denote the normal distribution with mean $\mu$ and standard deviation $\sigma$. We assume each feature to be normally distributed across all users with slight variations in mean and standard deviation across all features and users. More specifically, our methodology for generating the synthetic dataset is as follows.
\begin{enumerate}[leftmargin=5mm]
    \item We model the mean of all $n$ features taking values in the unit interval $\mathbb{I}$ as a normally distributed random variable $\mathcal{N}(\mu_{\text{mn}}, \sigma^2_{\text{mn}}) = \mathcal{N}(0.5, 0.1^2)$. Similarly we model the standard deviation of all $n$ features as another normally distributed random variable $\mathcal{N}(\mu_{\text{var}}, \sigma^2_{\text{var}}) = \mathcal{N}(0.1, 0.07^2)$. 
    \item For each feature $i \in [n]$, we first sample $\mu_i \leftarrow \mathcal{N}(\mu_{\text{mn}}, \sigma^2_{\text{mn}})$ and $\sigma_i \leftarrow \mathcal{N}(\mu_{\text{var}}, \sigma^2_{\text{var}})$. The resulting normal distribution $\mathcal{N}(\mu_{i}, \sigma^2_{i}) $ serves as the population distribution of the mean of the feature $i$.
    \item For each user $u$, we sample the mean $\mu_{u, i} \leftarrow \mathcal{N}(\mu_{i}, \sigma^2_{i})$. The variance $\sigma^2_{u, i}$ is chosen as the control variable. User $u$'s samples for the $i$th feature are generated as i.i.d. random variables $\mathcal{N}(\mu_{u, i}, \sigma^2_{u, i})$, which serves as user $u$'s distribution for the $i$th feature. 
\end{enumerate}

We evaluate the same four types of ML architectures, LinSVM, RBFSVM, RNDF and DNN. Due to the large number of potential configurations we evaluate the model performance at a fixed threshold of 0.5. For the experiments we choose 50 (synthetic) users, with 50 features in the feature space. Each experimental run is repeated 50 times \blue{to reduce any potential biases arising from the random process.}

\blue{
\begin{figure*}[th!]
	\centering
    \includegraphics[width=1.0\textwidth]{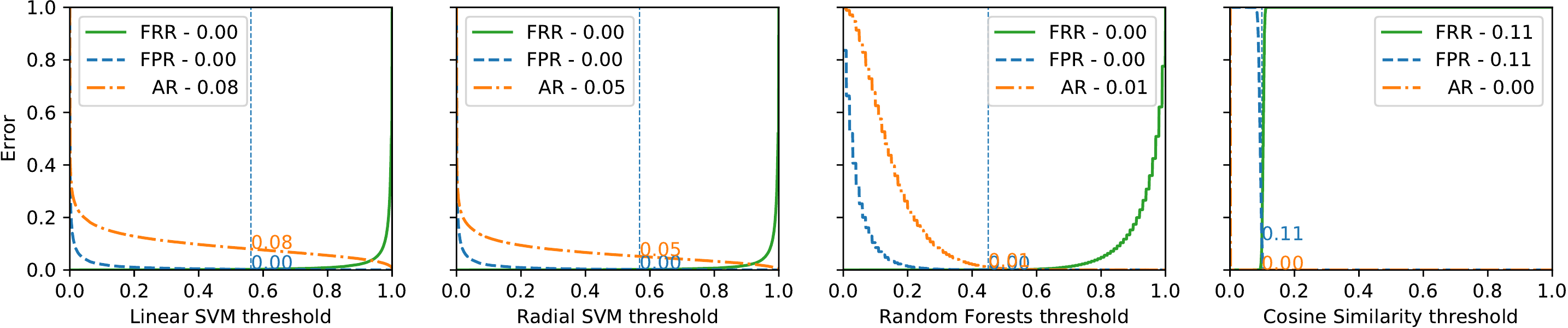}
    \vspace{-0.7cm}
  	\caption{\blue{ROC Curves versus the AR curve for different ML architectures, including a cosine similarity distance-based classifier. Trained on synthetic data of 50 features of 50 user, with fixed mean and variance for features of all users, repeated 50 times.}}
 	\label{fig:cosine_results}
 	\vspace{-0.1cm}
\end{figure*}

\begin{figure*}[th!]
	\centering
    \includegraphics[width=1.0\textwidth]{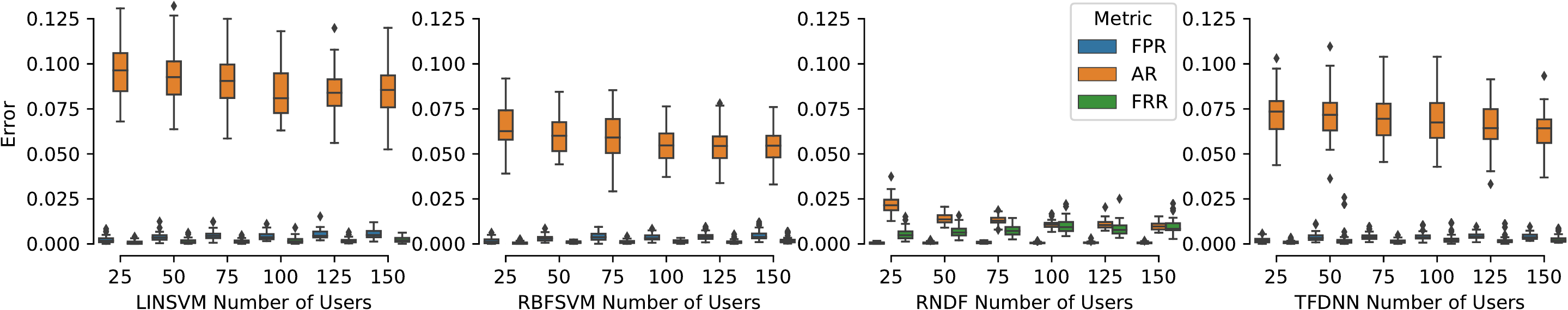}
    \vspace{-0.6cm}
  	\caption{\blue{A comparison between FPR and AR of four different ML architectures. Trained on synthetic data of 50 features per user, with a variable number of users, repeated 50 times.}}
 	\label{fig:vary_users}
 	\vspace{-0.2cm}
\end{figure*}
}

\subsection{Effects of Feature Variance on Acceptance Region}
\label{sub:eff-var}

\subsubsection{Variable Isolated User Variance and Fixed Population Variance}

We first treat one out of the 50 users as an outlier, whcih we call the \emph{isolated} user. 
The variance $\sigma^2_{u, i}$ is fixed at $(0.2)^2$ for all other users $u$ and for all features $i \in [n]$. We vary the variance $\sigma_{u_\text{tgt}, i}$ of the isolated user $u_\text{tgt}$ from 0.05 to 0.35 in increments of 0.05. Figure~\ref{fig:isolate_userv} plots the user's standard deviation $(\sigma_{u_\text{tgt}, i})$ relative to the fixed population standard deviation $(\sigma_{u, i})$ of 0.2. 
It is clear the overall AR, FRR and FPR of the users is not affected by changing feature variance of a single user, despite the isolated user's samples included as part of training and testing data of other users. Conversely, when viewing the AR, FRR and FPR of the isolated user, we observe a slight increase in FRR and FPR as the relative variance increases. This is due to the positive samples being spread out due to increased variance in the isolated user's samples. However, this is accompanied by a substantially large increase in the acceptance region of this user, approaching 1, i.e., the entire feature space. Furthermore, this trend is visible for all four classifiers.

\subsubsection{Fixed Isolated User Variance and Variable Population Variance}


In this experiment, we fix the variance $\sigma^2_{u_\text{tgt}, i}$ of the isolated user ($u_\text{tgt}$) at $(0.2)^2$. The $\sigma^2_{u, i}$ of the remaining population is sampled from a normal distribution $\sigma_{u, i} \leftarrow \mathcal{N}(\mu_{i}, \sigma^2_{i})$.
Where $\mu_{i}$ and $\sigma_{i}$ is sampled from the following distributions
$\mathcal{N}(\mu_{\text{mn}}, \sigma^2_{\text{mn}}) = \mathcal{N}(\mu_{\text{mn}}, 0.05^2)$ and $\mathcal{N}(\mu_{\text{var}}, \sigma^2_{\text{var}}) = \mathcal{N}(0.03, 0.02^2)$, respectively. $\mu_{\text{mn}}$ is varied between 0.05 and 0.35 in increments on 0.05
%
This sampling permits a small amount of variation between features.

The results are shown in Figure \ref{fig:isolate_otherv}. Inspecting the average AR, FRR and FPR of the system, it is evident there is a continual increase of all 3 metrics as the relative variance increases. This increase is expected as the majority of users' feature values have high variance, presenting an increasingly difficult problem for the machine learner to reduce misclassification errors. However, in all four classifiers the average AR curve is either comparable or lower than the FPR curve as the relative variance increases. For the isolated user, we see that when the relative variance of all other users is lower than this user (to the left), the AR is significantly higher even though the FPR and FRR are minimal in all four classifiers. This shows that less variance in the population samples will result in a high AR, as the classifier need not tighten AR around the true positive region, due to lack of high variance negative samples. On the other hand, AR of the isolated user decreases as the relative variance of the population increases.

\blue{
\subsection{On Distance Based Classifiers}
\label{sec:distance_classifier}
}

\blue{
As noted earlier, it has been stated that random inputs are ineffective against distance-based classification algorithms~\cite{pagnin2014leakage}. This is in contrast to the machine learning based algorithms evaluated in this paper. 
We take a brief interlude to experimentally evaluate this claim on the cosine similarity distance-based classifier. 
We sample 50 features with means distributed as 
$\mathcal{N}(\mu_{\text{mn}}, \sigma^2_{\text{mn}}) = \mathcal{N}(0.2, 0.05^2)$ and variance distributed as $\mathcal{N}(\mu_{\text{var}}, \sigma^2_{\text{var}}) = \mathcal{N}(0.03, 0.02^2)$. Cosine similarity is computed between two vectors of the same length. As our positive training data contains more than one training sample, we use the average of these samples as the representative template of the user~\cite{chauhan2016gesture}. We use a fixed number of 50 users, with the experiment repeated 50 times. Recall that our evaluation at each threshold is best-effort; we use 1,000 threshold bins for the evaluation of the cosine similarity classifier, since the FRR and FPR rapidly change over a small range of thresholds.

Figure~\ref{fig:cosine_results} displays three classical machine learning algorithms of linear SVM, radial SVM, and random forests, alongside a distance-based cosine similarity classifier. 
It is clear from the figure, that the AR is near zero for cosine similarity, unlike the other classifiers using the same synthetic dataset. This, however, comes at the cost of higher EER. This suggests that distance-based classifiers are effective in minimizing the AR of model, but at the expense of accuracy of the system. We leave further investigation of distance-based classifiers as future work.

\subsection{Effects of Increasing Synthetic Users}
\label{sec:vary_users}
The real-world datasets used in Section~\ref{sec:realscheme} have a variable number of users. Our binary classification task aggregates negative user samples into a negative class, resulting in distributions and variances of the negative class which depend on the number of users in the datasets. Thus, in this test we investigate the impact on TPR, FPR and AR by varying the number of users in the dataset. We use the synthetic dataset configured in the same manner as in Section~\ref{sec:distance_classifier}. We increase the number of users within the synthetic dataset, from 25 to 150, in increments of 25. Note that the split between positive and negative samples is still balanced (see Remark~\ref{rem:balanced}).

In Figure \ref{fig:vary_users}, we observe that with the addition of more users, there is a slight increase in the FPR. This is expected as the likelihood of user features being similar between any two users will increase with more users in the population. As the training of the classifier uses samples from other users as a negative class, the increased number of negative users slightly lowers the AR of the classifier, with an increased variation of the negative training set (from additional users) covering more of the feature space. However, both these changes are relatively minor despite the multi-fold increase in the number of users. Thus, the AR of the classifiers remains relatively stable with an increasing number of users.

}




\begin{table*}[t]
\caption{Equal Error Rate and AR
with and without the mitigation strategy. Green (resp., red) shades highlight improvement (resp., deterioration) in FPR and AR. Color intensity is proportional to degree of performance change.}
\label{tab:summary_ar}
\vspace{-0.1cm}
\centering
\resizebox{0.80\textwidth}{!}{%
\begin{tabular}{|l||l|l|l|l||l|l|l|l||l|l|l|l||l|l|l|l|}
\hline
 & \multicolumn{4}{c||}{Linear SVM} & \multicolumn{4}{c||}{Radial Svm} & \multicolumn{4}{c||}{Random Forest} & \multicolumn{4}{c|}{Deep Neural Network} \\ \hline
Biometric & \multicolumn{2}{c|}{Normal} & \multicolumn{2}{c||}{Mitigation} & \multicolumn{2}{c|}{Normal} & \multicolumn{2}{c||}{Mitigation} & \multicolumn{2}{c|}{Normal} & \multicolumn{2}{c||}{Mitigation} & \multicolumn{2}{c|}{Normal} & \multicolumn{2}{c|}{Mitigation} \\ \cline{2-17}
Modality & FPR & AR & FPR & AR & FPR & AR & FPR & AR & FPR & AR & FPR & AR & FPR & AR & FPR & AR \\ \hline
Gait & 0.160 & 0.24 & \cellcolor{red!0} 0.160 & \cellcolor{green!25} 0.04 & 0.140 & 0.18 & \cellcolor{red!0} 0.140 & \cellcolor{green!17.5} 0.04 & 0.09 & 0.03 & \cellcolor{red!0} 0.09 & \cellcolor{green!3.75} 0.00 & 0.215 & 0.20 & \cellcolor{green!5.625} 0.170 & \cellcolor{green!25} 0.00 \\
Touch & 0.325 & 0.49 & \cellcolor{red!1.875} 0.340 & \cellcolor{green!60} 0.01 & 0.265 & 0.41 & \cellcolor{red!0} 0.265 & \cellcolor{green!47.5} 0.03 & 0.21 & 0.23 & \cellcolor{red!0} 0.21 & \cellcolor{green!28.75} 0.00 & 0.325 & 0.30 & \cellcolor{red!6.25} 0.375 & \cellcolor{green!37.5} 0.00 \\
Face & 0.050 & 0.15 & \cellcolor{red!1.875} 0.065 & \cellcolor{green!5} 0.11 & 0.040 & 0.01 & \cellcolor{red!0} 0.040 & \cellcolor{red!0} 0.01 & 0.03 & 0.78 & \cellcolor{red!0} 0.03 & \cellcolor{green!97.5} 0.00 & 0.095 & 0.10 & \cellcolor{green!3.75} 0.065 & \cellcolor{green!7.5} 0.04 \\
Voice & 0.030 & 0.08 & \cellcolor{red!0} 0.030 & \cellcolor{green!2.5} 0.06 & 0.020 & 0.00 & \cellcolor{red!0} 0.020 & \cellcolor{red!0} 0.00 & 0.04 & 0.01 & \cellcolor{red!0} 0.04 & \cellcolor{green!1.25} 0.00 & 0.115 & 0.08 & \cellcolor{green!3.125} 0.090 & \cellcolor{green!7.5} 0.02\\ \hline
\end{tabular}
}
\end{table*}

\begin{table*}[t]
\caption{Equal Error Rate and RAR
with and without the mitigation strategy.
The AR values remain the same as in Table~\ref{tab:summary_ar}.
$\beta$-RAR 
indicates RAR treated with only $\beta$ noise. RAR 
indicates the inclusion of both $\beta$ noise and raw random input samples.}
\label{tab:summary_rar}
\vspace{-0.1cm}
\resizebox{\textwidth}{!}{%
\begin{tabular}{|l||l|l|l|l|l||l|l|l|l|l||l|l|l|l|l||l|l|l|l|l|}
\hline
 & \multicolumn{5}{c||}{Linear SVM} & \multicolumn{5}{c||}{Radial Svm} & \multicolumn{5}{c||}{Random Forest} & \multicolumn{5}{c|}{Deep Neural Network} \\ \hline
Biometric & \multicolumn{2}{c|}{Normal} & \multicolumn{3}{c||}{Mitigation} & \multicolumn{2}{c|}{Normal} & \multicolumn{3}{c||}{Mitigation} & \multicolumn{2}{c|}{Normal} & \multicolumn{3}{c||}{Mitigation} & \multicolumn{2}{c|}{Normal} & \multicolumn{3}{c|}{Mitigation} \\ \cline{2-21}
Modality & FPR & RAR & FPR & $\beta$-RAR & RAR & FPR & RAR & FPR & $\beta$-RAR & RAR &  FPR & RAR & FPR & $\beta$-RAR & RAR &  FPR & RAR & FPR & $\beta$-RAR & RAR \\ \hline
Touch Raw & 0.325 & 0.45 & \cellcolor{red!2.5} 0.345 & \cellcolor{green!1.25} 0.44 & \cellcolor{green!56.25} 0.00 & 0.265 & 0.40 & \cellcolor{red!0} 0.265 & \cellcolor{green!5} 0.36 & \cellcolor{green!48.75} 0.01 & 0.21 & 0.18 & \cellcolor{red!0.625000000000001} 0.215 & \cellcolor{green!16.25} 0.05 & \cellcolor{green!22.5} 0.00 & 0.325 & 0.32 & \cellcolor{red!6.875} 0.38 & \cellcolor{green!7.5} 0.26 & \cellcolor{green!40} 0.00 \\
Face Raw & 0.050 & 0.12 & \cellcolor{red!3.125} 0.075 & \cellcolor{red!2.5} 0.14 & \cellcolor{green!15} 0.00 & 0.040 & 0.09 & \cellcolor{red!0} 0.040 & \cellcolor{red!0} 0.09 & \cellcolor{green!11.25} 0.00 & 0.03 & 0.02 & \cellcolor{red!0} 0.030 & \cellcolor{green!1.25} 0.01 & \cellcolor{green!2.5} 0.00 & 0.095 & 0.10 & \cellcolor{green!3.125} 0.07 & \cellcolor{green!5} 0.06 & \cellcolor{green!8.75} 0.03\\ \hline
\end{tabular}
}
\vspace{-0.2cm}
\end{table*}

\section{Mitigation}
\label{sec:mitigation}
In the previous section, we validated that higher variance in the samples in the negative class as compared to the variance of samples from the target user class reduces AR. The data from the negative class is obtained from real user samples, and therefore scheme designers cannot control the variance. However, this gives us a simple idea to minimize AR: generate noise vectors around the target user's vectors and treat it as part of the negative class for training the model. This will result in the tightening of the acceptance region around the true positive region. We remark that the noise generated is independent of the negative training samples.





\subsection{The Beta Distribution}
More specifically, we generate additional negative training samples by sampling noisy vectors where each feature value is sampled from a beta distribution. We generate samples equal to the number of samples in the positive class. Thus creating a dataset with a third of the samples as positive, another third as negative samples from other users, and finally the remaining third of feature vectors treated as negative samples from the beta distribution dependent on the positive user. The procedure is as follows. For the $i$th feature, let $\mu_i$ denote the mean value for the given target user. We use the beta distribution with parameters $\alpha_i = | 0.5 - \mu_i| + 0.5$ and $\beta_i = 0.5$. We denote the resulting beta distribution by $\mathcal{B}e(\alpha_i, \beta_i)$. Then a noisy sample $\mathbf{x}$ is constructed by sampling its $i$th element $x_i$ from the distribution $ \mathcal{B}e(\alpha_i, \beta_i)$ if $\mu_i \le 0.5$, and from $1 - \mathcal{B}e(\alpha_i, \beta_i)$ otherwise. The two cases ensure that we add symmetric noise as the mean moves over to either side of $0.5$.




\descr{Results on AR.} In Table \ref{tab:summary_ar}, we show the resulting FPR and AR after the addition of beta noise at the equal error rate. The detailed ROC curves are shown in Figure~\ref{fig:all_beta_def} in Appendix~\ref{sec:mitigation_roc}. In every configuration (classifier-dataset pairs), we see a significant decrease in AR. The AR is now lower than FPR in every configuration. In 14 out of 16 cases, the AR is $\le 0.04$. The two exceptions are LinSVM (with face and voice datasets). We further see that in 13 out of 16 instances the FPR either remains unchanged or improves! The 3 instances where the FPR degrades are LinSVM with face and face datasets both by $+0.015$, and DNN with Touch where the difference is $+0.05$. Thus, adding beta distributed noise does indeed decrease the AR with minimal impact on FPR. This agrees with our postulate that high AR was likely due to loose decision boundaries drawn by the classifier, and the addition of beta noise tightens this around the true positive region. Figure \ref{fig:individual_ar-fpr_def} in Appendix~\ref{sec:mitigation_roc} displays individual user FPRs and ARs.

\descr{Results on RAR.} 
Interestingly, beta distributed noise only marginally reduces the raw acceptance rate as can be seen in Table~\ref{tab:summary_rar} (columns labeled $\beta$-RAR). 
The reason for this lies in the volume of the region spanned by random raw inputs. We previously saw in Section~\ref{sub:rar-results} that it was (a) exponentially small and (b) many orders of magnitude smaller than the true positive region. Thus, it is unlikely that beta distributed noise will lie in this region to aid the model to label them as negative samples. Consequently we sought another means to mitigate this attack surface. 

\subsection{Feature Vectors from Raw Inputs as Negative Samples}
Our mitigation strategy to reduce RAR is to include a subset of raw input vectors in the training process, whose cardinality is equal to the number of positive user samples in the training dataset. The training dataset now contains 1/4th each of raw input vectors, beta-noise, positive samples, and samples from other users.   


\descr{Results on AR and RAR.}
Table~\ref{tab:summary_rar} shows that the mitigation strategy reduces the RAR to less than or equal to 0.03 in all instances (columns labeled RAR). The resulting FPR is marginally higher than the FPR from only beta-distributed noise in some cases (Table~\ref{tab:summary_ar}). Thus, the inclusion of beta-distributed noise in conjunction with subset of raw inputs in the training data reduces both AR and RAR with minimal impact on FPR and FRR.

\section{Discussion}
\label{sec:discussion}
\begin{itemize}[leftmargin=5mm]
\item Our work proposes an additional criterion to assess the security of biometric systems, namely their resilience to random inputs. The work has implications for biometric template protection~\cite{biometric-template}, where a target template resides on a remote server and the attacker's goal is to steal the template. In such a setting, obtaining an accepting sample may be enough for an attacker, as it serves as an approximation to the biometric template. Our work shows that the attacker might be able to find an approximation to the template via random input attacks if the system AR is not tested. Conversely, once the AR is reduced below FPR (e.g., via adding beta distributed noise), then one can safely use FPR as the baseline probability of success of finding an approximation. 

\item We have assumed that the input to the classifier, in particular the length of the input is publicly known. In practice, this may not be the case. For instance, in face recognition, a captured image would be of a set size unknown to the attacker. Likewise, the number of features in the (latent) feature space may also be unknown. However, we do not consider this as a serious limitation, as the input length is rarely considered sensitive so as to be kept secret. In any case, the security of the system should not be reliant on keeping this information secret following Kerckhoffs's well known principle. 

\item We note that there are various detection mechanisms that protect the front-end of biometric systems. For example, spoofing detection \cite{marcel2014handbook} is an active area in detecting speaker style transfer \cite{wang2017tacotron}. Detection of replay attacks is also leveraged to ensure the raw captured biometric is not reused, for example audio recordings \cite{kinnunen2017asvspoof}. There is also liveliness detection, which seeks to determine if the biometric that is presented is characteristic of a real person and not a recreation, e.g., face masks remain relatively static and unmoving compared to a real face \cite{tan2010face}. Our attack surface applies once the front-end has been bypassed. Our mitigation measures can thus be used in conjunction with these detection mechanisms to thwart random input attacks. \blue{Being generic, our mitigation measures also work for systems which do not have defense measures similar to liveness detection.}






\item Once an accepting sample via the feature vector API has been found, it may be possible to obtain an input that results in this sample (after feature extraction), as demonstrated by 
Garcia et al. with the training of an auto-encoder for both feature extraction and the regeneration of the input image~\cite{garcia2018explainable}. 

\item \blue{In this work, we have focused on authentication as a binary classification problem, largely because of its widespread use in biometric authentication~\cite{chauhan2017behaviocog, curran2017one, huang2018breathlive, liu2018vocal, chen2017your, song2016eyeveri, chauhan2017breathprint, ho2017mini, crawford2017authentication, xu-soups}. However, authentication has also been framed as a one-class classification problem~\cite{one-class-face, xu-soups} or as multi-class classification~\cite{xu-soups}, e.g., in a discrimination model, as noted earlier. In one-class classification, only samples from the target user are used to create the template, and the goal is to detect outliers. If this is achieved in a manner similar to distance-based classifiers, then as we have seen in Section~\ref{sec:distance_classifier}, and as previously indicated in~\cite{pagnin2014leakage}, the AR is expected to be small. In the multi-class setting, each of the $n$ users is treated as a different class. This increase in classes is expected to proportionally lower the AR. However, whether this behavior is observed on real world data requires additional experimentation. We remark that as observed in Section~\ref{sub:eff-var}, AR is highly dependent on the relative variance of the positive user and the negative user features. This may lead to the possibility of larger AR for some of the users, consequently leading to higher risk of attack for these users. We leave thorough investigation of the one-class and multi-class settings as future work.}


\end{itemize}


\section{Related Work}
\label{sec:related}
There are several mentions of attacks similar to the random input attack discussed in this paper. Pagnin et al.~\cite{pagnin2014leakage} define a \emph{blind brute-force attack} on biometric systems where the attacker submits random inputs to find an accepting sample. The inputs are $n$-element vectors whose elements are integers in the set $\{0, 1, \ldots, q - 1\}$. The authors conclude that the probability of success of this attack is exponential in $n$, assuming that the authentication is done via a distance function (discarding any vector outside the ball of radius determined by the system threshold). They concluded that blind brute force attack is not effective in recovering an accepting sample. While this may apply to distance-based matching, the same conclusion cannot be made about machine learning based algorithms whose decision functions are more involved. Indeed, we have shown that the acceptance region for machine learning classifiers is not exponentially small. It has also been argued that the success rate of random input attacks can be determined by the false positive rate (FPR), at least in the case of fingerprint and face authentication~\cite{martinez-hill, vulnerable-face}. We have shown that for sophisticated machine learning classifiers this conclusion is not true, and random input attacks in many instances success at a rate higher than FPR. A more involved method is  hill-climbing~\cite{soutar-hill, martinez-hill} which seeks an accepting sample via exploiting the confidence scores returned by the matching algorithm. The authentication systems considered in this paper do not return confidence scores.

Serwadda and Phoha~\cite{serwadda2013kids} use a robotic finger and population statistics of touch behavior on smartphones to launch a physical attack on touch-based biometric authentication systems. Their attack reduces the accuracy of the system by increasing the EER. In contrast, our work does not assume any knowledge of population biometric statistics, e.g., population distribution of feature space. It is an interesting area of work to investigate whether a robotic finger can be programmed to generate raw inputs used in our attack. 


Garcia et al.~\cite{garcia2018explainable} use explainable-AI techniques~\cite{explainable-AI} to construct queries (feature vectors) to find an accepting sample in machine learning based biometric authentication systems. On a system with 0 FPR, they show that their attack is successful in breaching the system with up to 93\% success rate. However, their attack is more involved: it requires the construction of a seed dataset containing representative accepting and rejecting samples of a user set chosen by the adversary. 
This dataset trains a neural network as a substitute to the classifier of the authentication system.
The adversary then uses explainable AI techniques to obtain an accepting sample of a target user (not in the seed dataset) in as few queries as possible, by updating the substitute network. 
The authors also report a random feature vector attack, however, the attack is only successful on one out of 16 victims. The random feature vector is constructed by sampling each feature value via a normal distribution (distribution parameters not stated), unlike the uniform distribution in our case. We also note that they propose including images with randomly perturbed pixels as a counter-measure to defend against the aforementioned random input attack. This is different from our proposed beta-distributed noise mitigation technique, as it is agnostic to the underlying biometric modality.

The frog-boiling attack~\cite{chan2011frog, huang2011adversarial} studies the impact of gradual variations in training data samples to manipulate the classifier decision boundary. 
In this work we do not consider the adversary with access to the training process, nor do we evaluate models with an iterative update process. If this threat model is considered for the problem addressed in this paper, then an adversary may seek to maximize the acceptance region of a model by gradually poisoning the training dataset. As we have demonstrated in Section~\ref{sec:synthetic}, the relative variance between the user's data and population dataset directly impacts AR. Thus the manipulation of a user's training samples to be more varied would be effective in increasing the AR. Likewise, in our mitigation technique, we have shown that beta-distributed noise is effective in the minimization of AR. However an adversary might poison the training data by labeling beta noise as positive samples resulting in a maximization of the acceptance region to near 100\% of the feature space. 

Our work is different from another line of work that targets machine learning models in general. For instance, the work in~\cite{papernot-practical} shows an \blue{\emph{evasion attack} where the adversary, through only blackbox access to a neural network, forces the classifier to misclassify an input by slightly perturbing the input even though the perturbed sample is perceptually similar to the original sample, e.g., noisy images.}
The attack can be applicable to the authentication setting as well. However, it relies on the confidence values (probability vectors) returned by the classifier, which is not the case in authentication. Similarly, the work in~\cite{tramer2016stealing} shows how to steal a machine learning model, i.e., retrieve its undisclosed parameters, which only returns class labels (accept/reject decision in the case of authentication). They describe several techniques including the Lowd and Meek attack~\cite{lowd-meek} to retrieve a model sufficiently similar to the target model. The machine learning models considered in their attack are for applications different from authentication where one expects to find an accepting sample with negligible probability. 



\blue{There are also proposals to defend against the above mentioned evasion attacks. The goal is to make the classifiers \emph{robust} against adversarial inputs in the sense that classification is constant within a ball of certain radius around each input~\cite{madry2017towards, cohen2019certified}. Madry et al.~\cite{madry2017towards} propose a theoretical framework which formalizes defense against adversarial attacks by including adversarially perturbed samples in the loss function of DNNs. They show that it is possible to train DNNs robust against a wide range of adversarial input attacks. 
Cao and Gong~\cite{cao2017mitigating} propose another defense where given a test input, random points within a hypercube surrounding the input are sampled, and the majority label returned by the already trained DNN is assigned to the test input. 
Randomized smoothing~\cite{cohen2019certified} creates a separate classifier from any classifier such that its prediction within a Gaussian noise region (ball) around any input is constant, and consequently less likely to produce an erroneous prediction. We note that in evasion attacks there is a notion of \emph{nearness}, i.e., the adversary is given an input and seeks to add a small amount of noise such that the resultant erroneously labelled input is close to the original input. In contrast, in our case the random input need not be close to the target user's samples or even follow the same distribution. Furthermore, we have shown that even a conservative estimate of the true positive region is negligible in comparison to the entirety of the feature space (Section~\ref{sec:attack_area}). Thus, it is unclear whether such defenses apply to uniform random inputs, as opposed to random perturbations of inputs.}



Membership inference attacks \cite{shokri2017membership, salem2018ml} attempt to determine if a record obtained by an adversary was part of the original training data of the model. Whilst this attack does not compromise the security of the model, it breaches the privacy of the individual records. These attacks create a \emph{shadow model}~\cite{shokri2017membership} to mimic the behavior of the target model. Salem et al.~\cite{salem2018ml} construct a shadow model using only positive class samples and negative noise generated via uniformly random feature vectors. However it is hypothesized that these random samples belong to non-members, i.e., the negative class~\cite[\S V.B]{salem2018ml}. We have shown that a large portion of these random inputs may \blue{also belong to the positive class.} 


\blue{Finally, we point to other works in literature analyzing the security of biometric authentication systems. Sugrim et al.~\cite{sugrim2019robust} survey and evaluate a range of performance metrics used in biometric authentication schemes. They seek to motivate scheme designers to leverage robust metrics to provide a complete description of the system, including a proposal of the new metric: Frequency Count Score (FCS). The FCS metric shows a distribution of scores of legitimate and unauthorized users, identifying the overlap between the two distributions which helps to select the appropriate threshold for the classification decision. The FCS, however, is dependent on the negative class or samples of other users, which does not include random inputs. The work in~\cite{sugrim2019recruit} investigates the accuracy of authentication systems reported on a small number of participants when evaluated over an increasing number of users. The authors suggest that performance limits of a system with a small number of participants should be evaluated iteratively by increasing the participant count until a the performance degrades below a tolerable limit.  
}

\section{Conclusion}
It is important to assess the security of biometric authentication systems against random input attacks akin to the security of passwords against random guess attacks. We have demonstrated that without intentionally including random inputs as part of the training process of the underlying machine learning algorithm, the authentication system is likely to be susceptible to random input attacks at a rate higher than indicated by EER. Absent any other detection mechanism, e.g., liveliness detection, this renders the system vulnerable. The mitigation measures proposed in this paper can be adopted to defend against such attacks.

\section{Acknowledgments}
This research was funded by the Optus Macquarie University Cybersecurity Hub, Data61 CSIRO and an Australian Government Research Training Program (RTP) Scholarship. We would like to thank the anonymous reviewers and our shepherd Kevin Butler for their feedback to improve the paper. 



\begin{figure*}[!ht]
  \centering
  \subfloat[Touch Average ROC in the presence of Beta Noise\label{fig:touch_betarand_def}]{%
        \includegraphics[width=0.87\linewidth]{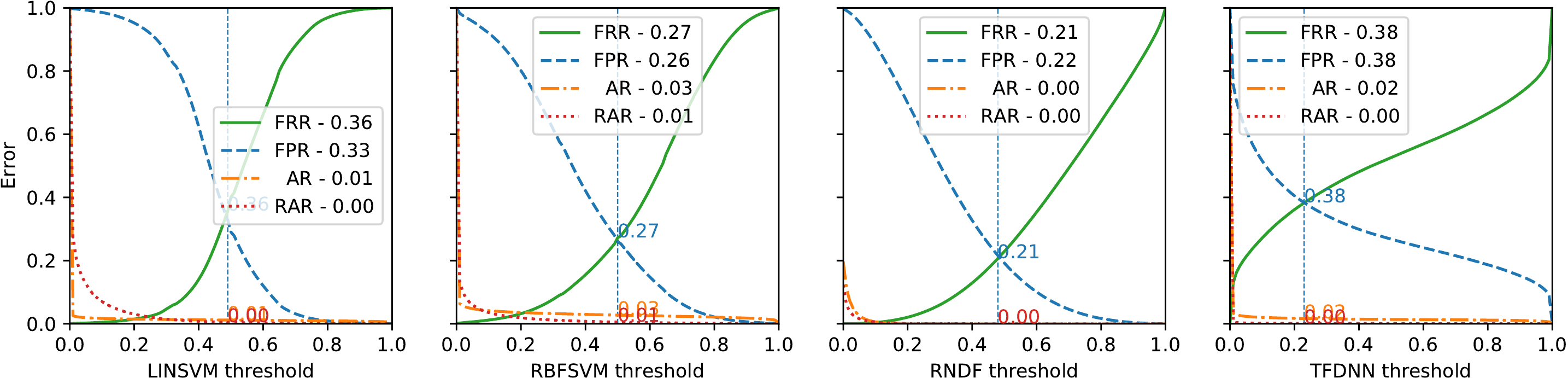}}\\[-0.2pt]
  \subfloat[Face Average ROC in the presence of Beta Noise\label{fig:face_betarand_def}]{%
        \includegraphics[width=0.87\linewidth]{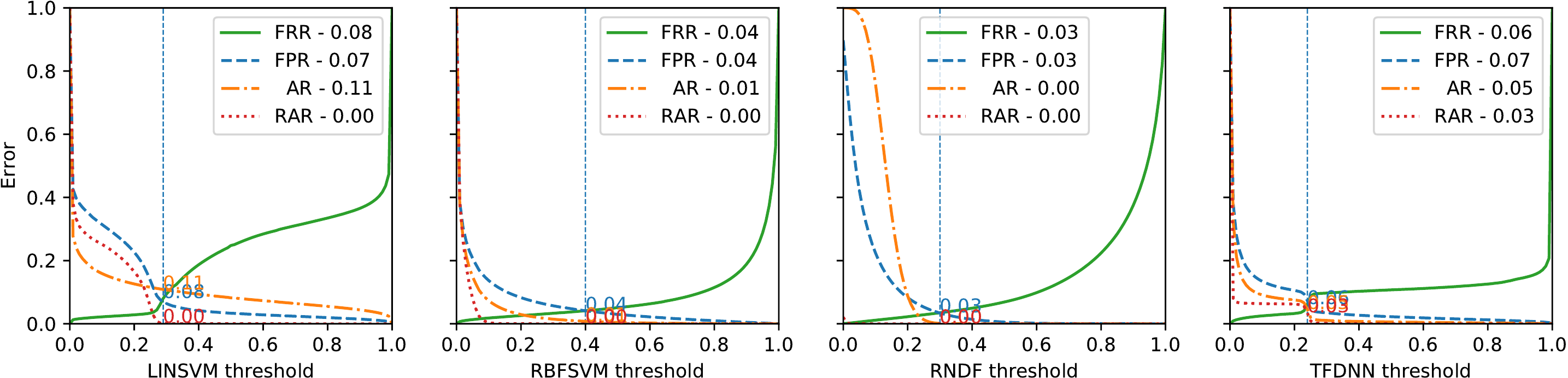}}\\
  \vspace{-0.2cm}
  \caption{Beta-noise mitigation of AR, with additional negative samples from the RAR feature set. The EER is marked on the diagrams as a vertical line. Addition RAR vectors were included as it was previously observed that beta noise is sufficient in mitigating AR attacks, but not the RAR attack.}
  \label{fig:all_betarand_def} 
  \vspace{-0.2cm}
\end{figure*}

\bibliographystyle{IEEEtran}
\bibliography{ref.bib}

\appendices


\vspace{-0.2cm}
\section{Mitigation ROC Plots}\label{sec:mitigation_roc}

This appendix contains plots of the results as discussed in Section \ref{sec:mitigation}.
Figure~\ref{fig:individual_ar-fpr_def} contains per-user scatter plots of AR and FPR for all biometric modalities and algorithms. For the same classifiers, Figure~\ref{fig:all_beta_def} illustrates the ROC curves for classifiers trained with the inclusion of beta distributed noise only. Finally, Figure~\ref{fig:all_betarand_def} displays the ROC curves for all classifiers of touch and face datasets with the inclusion of both beta distributed noise and raw input vectors as an additional mitigation strategy against the the raw inputs, which were unfazed by the beta noise. A summary of changes in FRR, FPR, AR and RAR of both Figure~\ref{fig:all_beta_def} and \ref{fig:all_betarand_def} have been provided earlier in Table~\ref{tab:summary_ar} and \ref{tab:summary_rar} of Section~\ref{sec:mitigation}.

\begin{figure*}[htbp]
    \centering
  \subfloat[Gait\label{fig:gait_ind_def}]{%
       \includegraphics[width=0.23\linewidth]{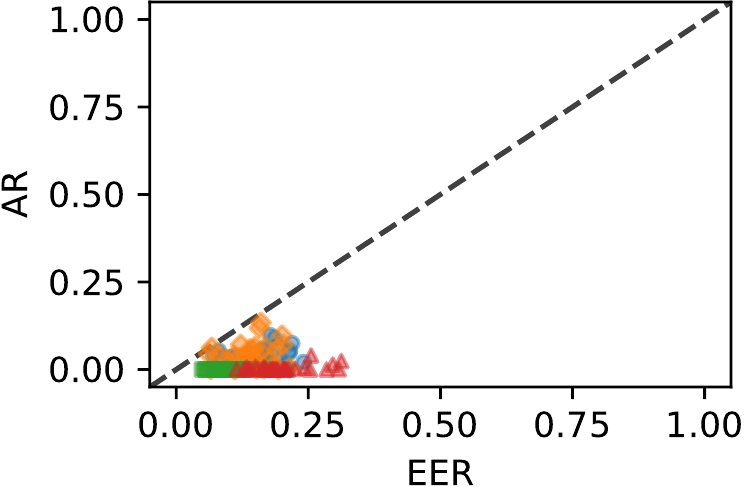}}
    \hspace*{0.7\fill}
  \subfloat[Touch\label{fig:touch_ind_def}]{%
        \includegraphics[width=0.23\linewidth]{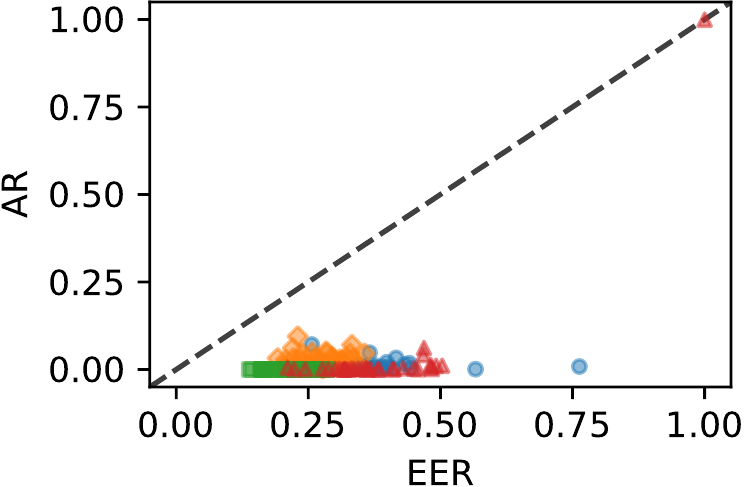}}
    \hspace*{0.7\fill}
  \subfloat[Face\label{fig:face_ind_def}]{%
        \includegraphics[width=0.23\linewidth]{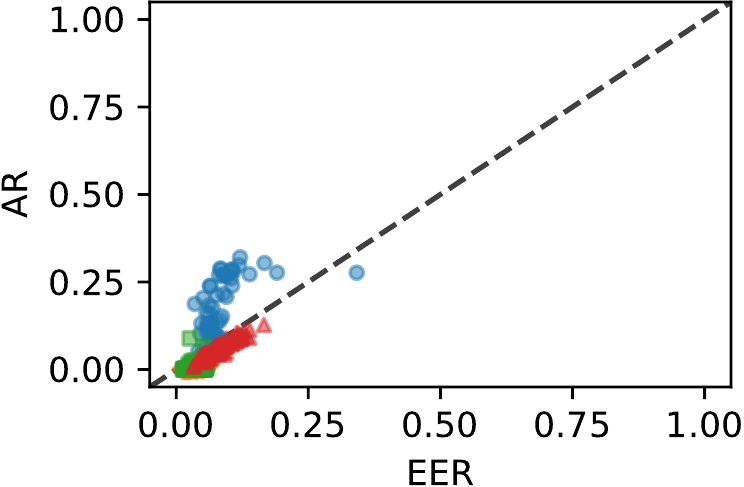}}
    \hspace*{0.7\fill}
  \subfloat[Voice\label{fig:voice_ind_def}]{%
        \includegraphics[width=0.23\linewidth]{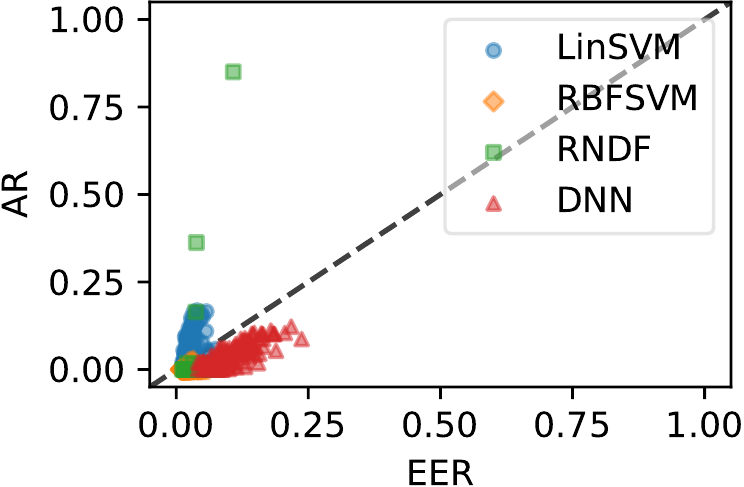}}
    \vspace{-0.2cm}
    \\
  \caption{Individual user scatter of AR and FPR after the addition of beta distributed noise. A substantial proportion of users now exhibit an AR close to zero, or below the AR = FPR. Unfortunately, this defense mechanism did not completely minimize the AR of LINSVM for the Face authenticator. Nor did this defense protect two outlying users in the RNDF voice authenticator.}
  \label{fig:individual_ar-fpr_def} 
  \vspace{-0.2cm}
\end{figure*}

\vfill

\begin{figure*}[htbp]
  \centering
  \subfloat[Gait Average ROC in the presence of Beta Noise\label{fig:gait_beta_def}]{%
        \includegraphics[width=0.88\linewidth]{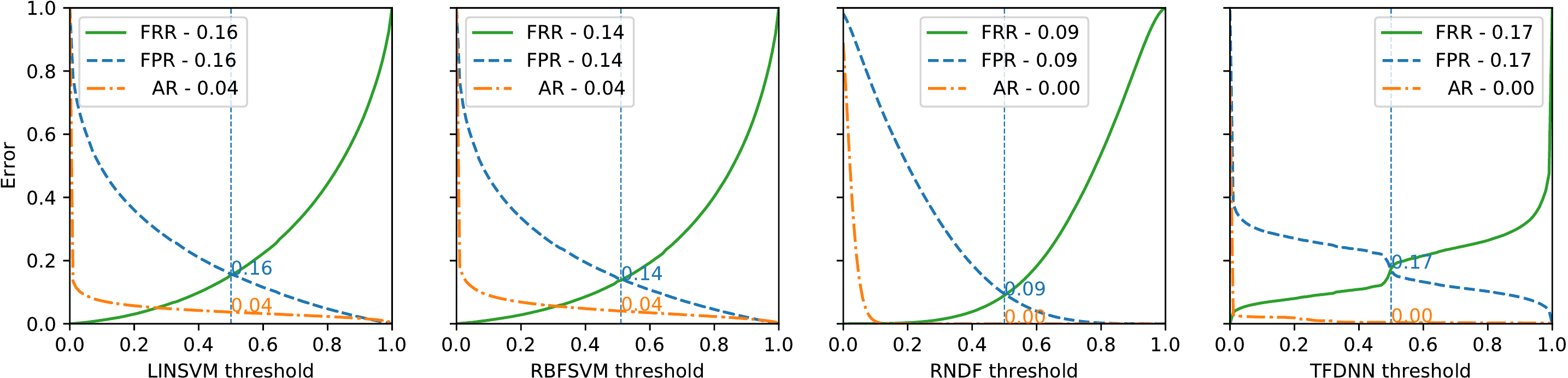}}\\
  \subfloat[Touch Average ROC in the presence of Beta Noise\label{fig:touch_beta_def}]{%
        \includegraphics[width=0.88\linewidth]{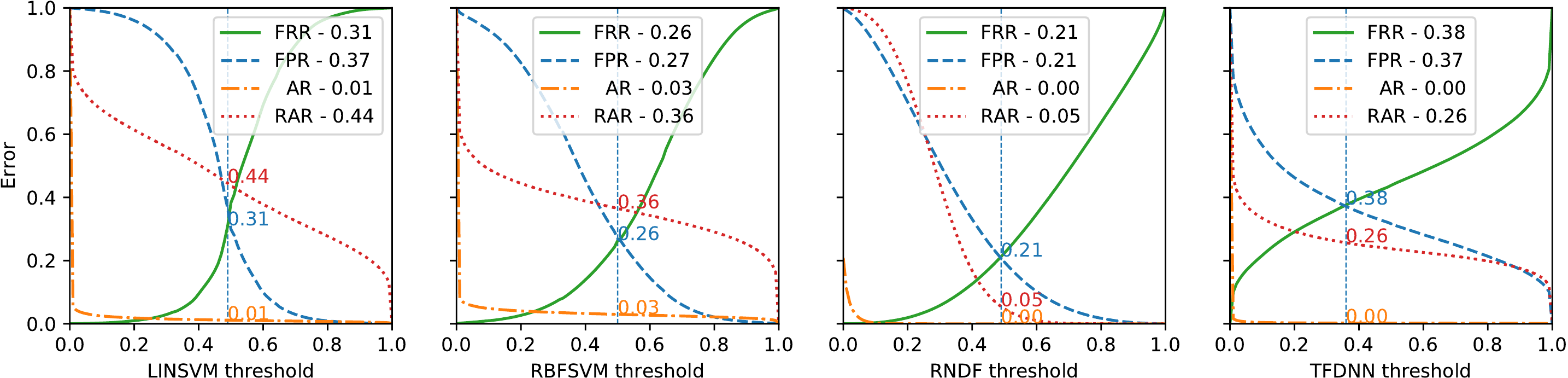}}\\
  \subfloat[Face Average ROC in the presence of Beta Noise\label{fig:face_beta_def}]{%
        \includegraphics[width=0.88\linewidth]{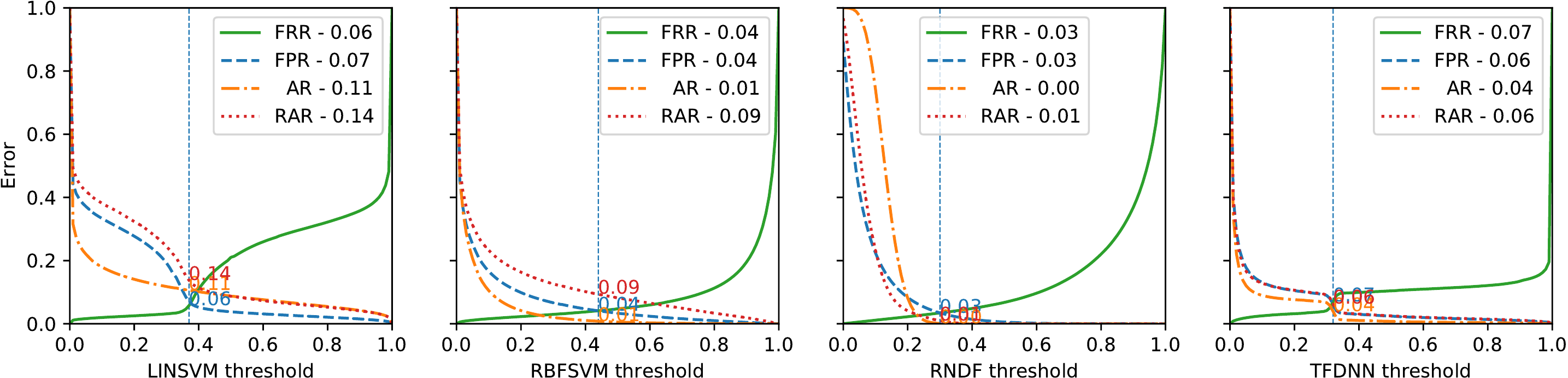}}\\
  \subfloat[Voice Average ROC in the presence of Beta Noise\label{fig:voice_beta_def}]{%
        \includegraphics[width=0.88\linewidth]{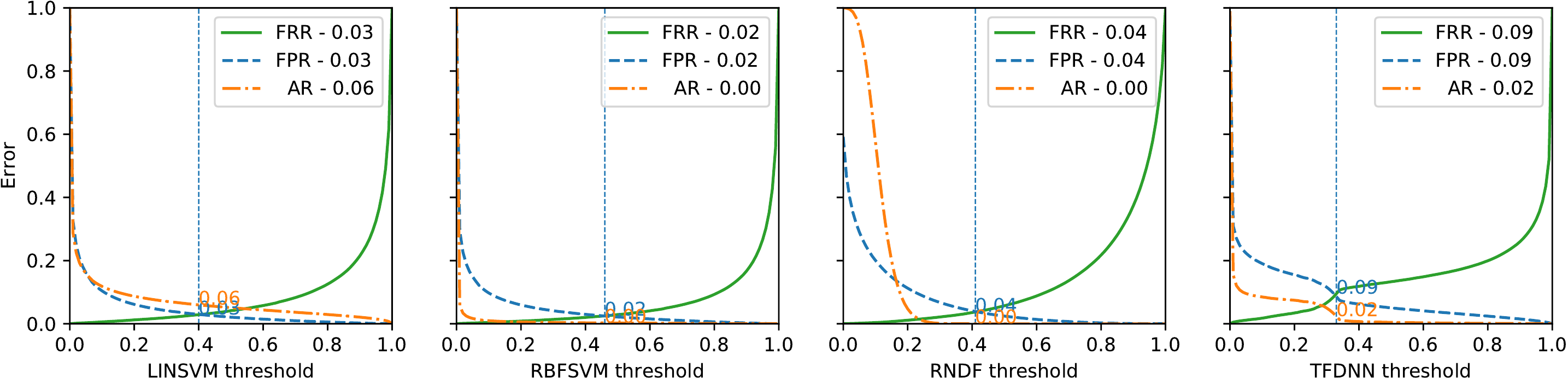}}
  \caption{Beta-noise mitigation of AR, with additive negative training noise sampled from a symmetric beta distribution around the mean of the user's features. The EER is marked on the diagrams as a vertical line. It is noted the plots with RAR curves the additional Beta-noise is not sufficient in mitigating RAR attacks.}
  \label{fig:all_beta_def} 
\end{figure*}

\section{DNN Estimator configuration.} \label{sec:dnn_config}

All models were trained for 5000 steps, with batch size of 50, with the Adagrad optimizer. \blue{The exact layer configuration of the \texttt{DNNEstimator}~\cite{tensorflow-estimators} used can be found on our project page (\url{https://imathatguy.github.io/Acceptance-Region/}).}



\end{document}